%% file: main.tex
\documentclass[a4paper, 12pt]{article}
\usepackage[margin=1in]{geometry}
\usepackage[affil-sl]{authblk}
\usepackage[utf8]{inputenc}
\usepackage{amsmath,amsthm,amsfonts,amssymb,amscd, color}
\usepackage{graphicx}
\usepackage[font=small]{caption}
\usepackage{subcaption}
\usepackage{todonotes}
\usepackage{natbib}
\usepackage{hyperref}
\usepackage{array}
\usepackage{threeparttable}
\hypersetup{colorlinks = true, urlcolor = blue, linkcolor = blue, citecolor= black}
\usepackage[ruled,vlined,noresetcount]{algorithm2e}
\usepackage{setspace}
\usepackage{booktabs}

\newtheorem{example}{Example}
\newtheorem{theorem}{Theorem}

\newtheorem{lemma}{Lemma}

\theoremstyle{remark}

\graphicspath{ {images/} }

\providecommand{\keywords}[1]
{
  \small	
  \textbf{\textit{Keywords:}} #1
}

\title{Nowcasting using regression on signatures}

\author[1,2]{Samuel N. Cohen}
\author[3]{Giulia Mantoan}
\author[2,4,5]{Lars Nesheim}
\author[4,6]{\'{A}ureo de Paula}
\author[7]{Arthur Turrell}
\author[8,9]{Lingyi Yang}
\affil[1]{University of Oxford}
\affil[2]{The Alan Turing Institute}
\affil[3]{Bank of England}
\affil[4]{University College London, Cemmap and Institute for Fiscal Studies}
\affil[5]{University of Bristol}
\affil[6]{CEPR}
\affil[7]{Google}
\affil[8]{University of Nottingham}
\affil[9]{Continuous Intelligence}
\date{\today}

\begin{document}

\maketitle

\begin{abstract}
\footnotesize
\noindent We introduce a new method of nowcasting using regression on path signatures. Path signatures capture the geometric properties of sequential data. Because signatures embed observations in continuous time, they naturally handle mixed frequencies and missing data. We prove theoretically, and demonstrate with simulations, that regression on signatures both subsumes the linear Kalman filter and has desirable consistency properties. Nowcasting with signatures is more robust to disruptions in data series than previous methods, making it useful in stressed times (for example, during COVID-19). This approach is performant in nowcasting US GDP growth, and in nowcasting UK unemployment.
\end{abstract}


\vspace{1em}
{\noindent \textit{\scriptsize{Disclaimer: 
Any views expressed are solely those of the author(s) and do not represent those of the Bank of England or its policy committees. This paper should therefore not be reported as representing the views of the Bank of England. We are grateful to seminar and conference participants at the 37th IARIW General Conference (2022), 2022, 2023, 2024 ESCOE Conference, the 2022 CFE-CMStatistics Conference, 2023 SIAM Conference on Financial Mathematics and Engineering, the 2023 Royal Economic Society Conference, the 2023 Financial Econometrics workshop, the 2023 Workshop on Nowcasting at Paris School of Economics, the 10th International Congress on Industrial and Applied Mathematics (2023), the 2024 Conference on Real-Time Data Analysis, Methods, and Applications, the 2024 Econdat Conference, the 2024 Society for Nonlinear Dynamics and Econometrics Conference, The 2024 British Applied Mathematics Colloquium, 2024 ECONDAT.
The authors gratefully acknowledge the ONS--Turing Strategic partnership for funding.
S.N.C. was supported by the UKRI Prosperity Partnership Scheme (FAIR) under EPSRC Grant [EP/V056883/1].  A.de P. was supported by the ESRC Institute for the Microeconomic Analysis of Public Policy (Grant Ref: ES/T014334/1) and by the UK government’s Horizon Europe funding guarantee (Grant Ref: EP/X02931X/1).  L.Y. was supported by EPSRC [EP/S026347/1] and the Hong Kong Innovation and Technology Commission (InnoHK Project CIMDA). 
For the purpose of open access, the authors have applied a CC BY public copyright licence to any Author Accepted Manuscript (AAM) version arising from this submission.
}}}

\clearpage

\onehalfspacing

\input{introduction}
\input{nowcast_challenges}
\input{theory}
\input{sig_nowcast}
\input{simulation}
\input{us_gdp}
\input{lfs}

\section{Conclusion} \label{sec:conclusions}

The path signature captures geometric properties of sequential data that make for excellent features to use in nowcasting. Signature methods allow for missing data from mixed frequencies or irregular sampling -- issues often encountered in nowcasting -- by embedding observed data in continuous time. They can also capture nonlinear effects and replace the time series modellers' need to think deeply about data transformations. All of these properties mean signature methods are well-placed to deal with difficult time series problems in general, and nowcasting in particular.

In this paper we demonstrated the application of regression on signatures to nowcasting via what we call the signature method. We introduced the theory of path signatures and illustrated how they are computed, and we showed that regression in the signature space subsumes the linear Kalman filter, which is commonly used in the nowcasting literature. We also validated this through simulations, demonstrating that it is possible to almost replicate the optimal performance of the Kalman filter (when parameters are known) by using regression on signatures (fitting parameters based on data) in simulated experiments, including with irregular data.

We also applied the signature method to some relevant real-world case studies. The first is a well-known empirical exercise: nowcasting US GDP growth. We have shown that incorporating signature features from principal components can offer improvements compared with a dynamic factor model based on \cite{bok2018macroeconomic}. The best results are achieved when we compute signatures on top of the fitted dynamic factors from the baseline DFM, illustrating the value of computing signatures as features even in the presence of informative data series.

The second case study was nowcasting UK unemployment with mixed frequencies and irregularly available data series. We showed that the signature method performs extremely well relative to a benchmark nowcast, that it can give an early indication of the sign of the change in the outturn, and that the performance of the signature holds up even in difficult circumstances -- in this case, during the COVID-19 pandemic. Such difficult circumstances are when informing policymakers matters most and therefore we believe the signature method to be an important and useful contribution to nowcasting.

\clearpage

\newpage
\singlespacing
\bibliographystyle{apalike} 
\bibliography{refs}

\newpage
\appendix
\section*{Appendix}
\input{appendices/app_framework}
\input{appendices/app_simulation.tex}
\input{appendices/app_configs}
\input{appendices/app_usgdp_factors.tex}
\input{appendices/app_exp_details.tex}
\input{appendices/app_lfs}
\input{appendices/app_consistency}

\end{document}

%% file: introduction.tex
\section{Introduction}

Policymakers need up-to-date economic information to make decisions. However, the collection and compilation of the data underlying key economic indicators takes time. For example, the first estimate of UK monthly GDP is published approximately 6 weeks after the end of the month it pertains to.  
Alternative indicators that have a shorter publication lag, or that are released more frequently, offer the opportunity to glean information about the economy in close to real-time and so inform policy in a timely manner. For instance, the UK's Office for National Statistics publishes a range of weekly alternative indicators of economic activity covering card spending, retail footfall, and more.

Models that can harness these alternative indicators to produce estimates of \textit{current} economic activity -- whether it be gross domestic product (GDP), employment or other official statistics -- are known as nowcasting models, and their predictions are called nowcasts. As an example, a nowcasting model might combine a range of alternative indicators and output an estimate of the GDP figure for the current month, in contrast to a forecast, which looks ahead to activity which has not yet occurred.

The economic literature on nowcasting has grown dramatically in the past 20 years \citep[e.g. see][]{stock2002forecasting,GiannoneReichlinSmall2006,bok2018macroeconomic}, and has focused on three challenges that arise when attempting to use a large and varied set of contemporaneous indicators for real-time insights on the economy. The first is how we incorporate data into a model when there are missing observations from mixed or irregular sampling frequencies. The second is how we extend beyond linear time-invariant systems to allow for time-varying parameters or for nonlinearity.

Third and finally, there is a challenge in how to incorporate large numbers of predictors into a single model, i.e. imposing structure to reduce the dimensionality of the problem.

During roughly the same time period, a set of tools to model nonlinear multidimensional time series has been developed in the mathematics and machine learning literatures (see Section~\ref{sec:sigs} for further details). These tools model time series as continuous time ``paths" and extract low dimensional representations of these paths known as ``signatures". It has been shown that path signatures perform well in prediction tasks \citep{graham2013sparse, morrill2020utilization}, suggesting that signatures can provide a tool to address the three nowcasting challenges.

Building on these developments, in this paper, we introduce a new nowcasting method, regression on signatures, that addresses all three of the nowcasting challenges in new ways: regression on signatures is compatible with standard dimension reduction tools, handles missing observations and mixed frequencies, and requires minimal additional effort when inferring nonlinear relationships. Path signatures are compatible with dimension reduction tools because they are features extracted from data that can be treated like any other variables. They handle mixed frequencies and missing data because they effectively model time series in continuous time and capture the interactions between features in which the ordering of events is preserved \citep{gyurko2013extracting, levin2013learning}. Finally, signatures are able to infer nonlinear relationships because they are nonlinear features with universal representation properties \citep{lyons2007differential} and, as with state-space models, signatures can be used in non-stationary settings.

An additional key advantage of regression on signatures is the transparency of its estimation, particularly when compared to other next-generation nowcasting approaches that leverage deep learning. While machine learning can yield strong predictive performance, its underlying models are often high-dimensional and nonlinear, making it difficult to interpret estimated relationships or assess robustness. Recent advances (e.g. SHAP values \citep{lundberg2017unified}, counterfactual methods \citep{black2021consistent}) have improved our ability to explain complex models, but these remain secondary layers of analysis rather than part of the estimation itself and can lack the transparency and simplicity of regression. In contrast, signature-based models rely on linear regression, for which the estimation procedures, sampling properties, and interpretation of coefficients are well-understood \footnote{The signature terms themselves may be abstract, but they do have a geometric interpretation (see Section~\ref{sec:sigs}).}. The linear structure of regression on signatures ensures that estimation remains transparent, stable, and amenable to standard inference techniques.

Using the path signature for nowcasting is very simple. First, we take available observations and compute their signature (the data points are treated as discrete observations of continuous paths). This can be done using any one of several free open source Python libraries that compute signatures, for example, \texttt{esig}, \texttt{iisignature} \citep{iisignature}, and \texttt{roughPy} \citep{morley2024roughpy}. Second, the computed signatures are used as regressors to predict the target variable of interest. This will be referred to as the \textit{signature method} in the rest of the paper. The signature method can be combined with standard dimension reduction methods prior to the first step by e.g. extracting principal components before computing the signatures.

The introduction of the signature method for nowcasting is the primary contribution of this paper. This is a new approach to nowcasting and has the aforementioned advantages relative to existing methods. Additionally, we show that the signature method provides superior performance on a standard problem in the literature, nowcasting US GDP, when compared to a direct application of Dynamic Factor Models (including the New York Federal Reserve's Dynamic Factor Model described by \citet{bok2018macroeconomic}). We also show that the signature method performs strongly when nowcasting UK unemployment changes, despite the challenging context of irregular information flows, multiple delays in the publication of the input and target variables due to external events (such as the COVID-19 pandemic), and a range of observation frequencies.

Our second contribution is to prove mathematically that the Kalman--Bucy filter can be written as a linear function of signatures. Put another way, signatures generalise linear Kalman filters, and recover the Kalman filter as a special case. In addition to our theoretical proof, we provide simulations that confirm estimates from the signature method do indeed match those of the Kalman filter.

This is a remarkable result: regression on signatures can capture any state-space model, even those with nonlinear relationships, in any number of dimensions. Furthermore, using signatures as features reduces some of the traditional complexity of time series modelling: nonlinear transformations and interactions between predictors are captured automatically, reducing the need for manual feature engineering. Because of these properties, this simple approach is applicable to a wide range of nowcasting tasks. 
The price for this flexibility is that the signature method can be high-dimensional, resulting in a need for substantial quantities of data or the use of dimension-reduction tools to prevent overfitting.

Our third contribution is to show how to combine signatures with standard dimension reduction techniques, like principal component analysis and dynamic factor models, so that the approach we introduce can be used in as wide a variety of settings as possible and can make use of large numbers of indicators effectively.

This paper contributes to the more recent literature on new methodologies to nowcast macroeconomic variables, including \citet{jokubaitis2021sparse}, \citet{proietti2021nowcasting}, \citet{soybilgen2021nowcasting}, \citet{aastveit2018combined}, \citet{ankargren2021simulation}, \citet{babii2022machine}, and on modelling missing data points and dimension reduction, such as \citet{pellegrino2020selecting}, \citet{smeekes2021automated}, \citet{buccheri2021score}, and \citet{chan2023high} in a more general context than nowcasting.

The rest of the paper is organised as follows: Section~\ref{sec:nowcasting_challenges} gives an overview of some challenges that nowcasting methods may address; Section~\ref{sec:background} details nowcasting from a mathematical perspective using signatures; this includes background on discrete and continuous time state-space models, Kalman filters, and proves that the Kalman filter can be equivalently written as a linear regression on the signature space; Section~\ref{sec:sigregression} summarises how to use regression on signatures in nowcasting, in terms of the pipeline and the practical implementation; Section~\ref{sec:simulation} demonstrates that it is possible to replicate the performance of the Kalman filter by using regression on signatures in a simulated problem; Section~\ref{sec:NYfedapplication} applies the signature model to nowcast US GDP growth drawing comparisons with a dynamic factor model based on the New York Fed staff nowcast model; Section~\ref{sec:lfsapplication} showcases the ability of regression on signatures to handle mixed frequency and irregularly sampled time series by nowcasting the UK monthly unemployment rate and compares the performance against AR(1) and MIDAS; finally, Section~\ref{sec:conclusions} concludes.  

%% file: nowcast_challenges.tex
\section{Background on challenges in nowcasting} \label{sec:nowcasting_challenges}

In this section, we discuss commonly occurring challenges for nowcasting in economics in more detail. We will address these challenges using the signature method in later sections.

\subsection{Missing observations and mixed frequency or irregular sampling} \label{subsec:irregular}
Due to varying publication lags, delays, and interruptions, the time series for some variables entering a nowcasting model often have missing observations. When these missing observations affect the end of the observation window, this is often referred to as the ``ragged-edge'' problem \citep{wallis1986forecasting}. The challenge is how to adapt standard discrete-time time-series methods to condition on arbitrary patterns of missing information. Similarly, we would like to use all possible high-frequency variables (for example daily, weekly, or monthly data on prices, expenditure, or revenue) to predict low-frequency variables (such as monthly or quarterly GDP). This is the mixed frequency data problem. Again the challenge is how to combine these complex mixed frequency data.

Several solutions to these problems have been developed in the macroeconometrics literature. One option for the missing data problem is to use the classic discrete time Kalman filter \citep[e.g.][]{kalman1960new,stock2002forecasting,banbura2014maximum}. This approach assumes a linear time series structure and imputes missing values using the Kalman filter. We discuss this option in more detail in Section \ref{sec:background}.

A second option to address both problems is to combine the bridge method \citep{schumacher2016comparison} with imputation. The bridge method addresses the mixed frequency problem using a two-equation system, first aggregating the set of high-frequency predictors to produce low-frequency predictors and then using the low-frequency predictors to predict the low-frequency target. We illustrate how bridge methods work via an example. Suppose we wish to use a single monthly variable $Z_t$ to predict a quarterly target $Y_t$. Here, $t$ indexes the month and $Y_t$ is the quarterly variable corresponding to the quarter ending at month $t$ whereas $Z_t$ is the monthly variable for month $t$. One bridge model for this problem is
\begin{align}\label{eq:bridge}
    Y_t&=\sum_{s=1}^{p} \alpha_s Y_{t-3s} + \beta_0+\beta_1 X_t+\varepsilon_t \\ 
    X_t & = \sum_{s=0}^2 \gamma_s Z_{t-s}.
    \label{eq:bridge2}
\end{align}
Equation (\ref{eq:bridge}) is an autoregressive equation used to predict $Y_t$, where \(p\) indicates the lag in the autoregressive process of the quarterly target and \(\epsilon_t\) is an exogenous noise process. The parameters in Equation (\ref{eq:bridge}) are typically estimated from the data. Equation (\ref{eq:bridge2}) is the ``bridge" equation that aggregates the high-frequency variable to the same frequency as the low-frequency variable. The parameters in the bridge equation are typically chosen \textit{ex ante}. For example, if $X_t$ is the quarterly average of $Z_t$, then $\gamma_s = \frac{1}{3}$ for all $s$.

The bridge equation addresses the mixed frequency problem. To overcome the ragged edges, missing values of $Z_t$ are imputed using an auxiliary forecasting model such as an AR(p) process
\begin{equation*}
     Z_{t} = \delta_{0} + \sum_{s=1}^p \delta_s Z_{t-s} + \eta_{t},
\end{equation*}
where \(\eta_t\) is a zero mean white-noise process, \(\delta_0\) and \(\delta_s\) are parameters to be fitted.

For simple patterns of mixed frequency data (e.g. monthly predictors and quarterly targets), the bridge method is easy to use and does not require users to specify in detail the high frequency dynamics of the system. However, the aggregation method employed is ad hoc and likely needs to be adapted for different contexts. Moreover, adapting the method to more complicated patterns of mixed frequency data (e.g. combining weekly, monthly, or irregularly sampled data) can be challenging.

A third approach directly relates high-frequency indicators to a low-frequency target: this is the mixed data sampling (MIDAS) model \citep{ghysels2005there}. MIDAS accommodates data sampled at different frequencies by incorporating higher frequency variables using distributed lag terms with fractional lags. Mathematically, it can be written in a form that is very similar to the bridge equation; the distinction is that MIDAS employs a more general functional form than Equation (\ref{eq:bridge2}) and its parameters are estimated rather than specified \textit{ex ante}. An example of the MIDAS equation for a target $Y_{t}$ is:
\begin{equation}\label{eq:midas}
Y_t = \sum_{i=1}^{k_{lo}} a_i X^{\text{lo}}_{i,t} + \sum_{j=1}^{k_{hi}} b_j \sum_{k=0}^{m} w_j(k; \theta_j)\cdot X^{\text{hi}}_{j,t-k} + \varepsilon_t,
\end{equation}
where $Y_t$ is the dependent variable, $X^{\text{lo}}_{i,t}$ are low-frequency predictors with coefficients $a_i$, $X^{\text{hi}}_{j,t-k}$ are high-frequency predictors with coefficients $b_j$, $w_j(k; \theta_j)$ are exponential Almon weights \citep{almon1965distributed}, and $\varepsilon_t$ is the error term. The exponential Almon weight function is defined by
$$
w_j(k; \theta_j) = \frac{\exp(1 + s_k(\theta_{j1} + \theta_{j2}s_k))}{\sum_{k=0}^m \exp(1 + s_k(\theta_{j1} + \theta_{j2}s_k))},
$$
where $s_k = k/m$ is the normalised lag position, $\theta_{j1},\: \theta_{j2}$ are the weight parameters for predictor $j$, and $m$ is the maximum lag length. The exponential Almon lag avoids parameter proliferation when $k_{hi}$ in Equation (\ref{eq:midas}) is large.

The objective function used to estimate the coefficients in the MIDAS model is 
$$
\min_{{a_i, b_j, \theta_j}} \frac{1}{2}\sum_{t=1}^T (Y_t - \hat{Y}_t)^2,
$$
where $\hat{Y}_t$ is the predicted value and the optimisation is subject to the constraints $\theta_{j1}, \theta_{j2} \geq 0 \:\forall j$. Nonlinear optimisation solvers such as BFGS \citep{broyden1970convergence, fletcher1970new, goldfarb1970family, shanno1970conditioning} and L-BFGS-B \citep{byrd1995limited} can be used to find the coefficients.

To date, most empirical nowcasting studies using the methods above only deal with regularly sampled data with missing observations. This is because nowcasting typically uses official published data that are released on a regular schedule. However, in recent years there has been an increase in the use of alternative data sources, e.g. web-scraped data and supermarket scanner data. These alternative sources can be high-frequency, and have complicated or irregular missingness patterns. Furthermore, statistical offices have recently struggled with issues such as COVID-19 and declining survey response rates that have, at times, caused significant delays in publishing both input and target variables relevant to economic nowcasting. 

Therefore, there is a need for a nowcasting methodology that can handle mixed-frequency data of arbitrary frequencies, complex irregularly sampled data, and data with arbitrary patterns of missing values. Signature methods can handle all of these cases. We discuss these points in the sections below and illustrate our results with simulations and empirical applications.

\subsection{Time-varying parameters and nonlinearities} \label{subsec:nonlinear}

Classic time series models, such as the autoregressive integrated moving average (ARIMA) model, assume that the dynamics of a low-dimensional variable, after suitable differencing to ensure stationarity, can be modelled by a linear model with autoregressive and moving average components. That is, after differencing, the data are stationary, the model is linear, and the parameters of the model are assumed to be constant. In nowcasting settings, all three of these assumptions can be restrictive -- for example, when there are large shocks to the economy, which is also when it is critical to have an accurate guide to what is happening.

State-space models offer one approach to relax these restrictions \citep{durbin2012time}. While the linear Kalman filter is the best known, state-space models in general are highly flexible and can allow for nonlinearity and time-varying parameters. 
A common difficulty in using these methods is that the state-space representation needs to be specified, and cannot easily be learned from historical data.
Many papers, including \cite{hamilton1994state}, \cite{nielsen2014estimation}, \cite{carter1994gibbs}, \cite{kim1999state}, \cite{kim1994dynamic}, also allow for time-varying parameters.

\subsection{Dimension reduction} \label{subsec:dimension}

In most nowcasting exercises, a large number of indicators are used to predict a single target.
For example, the Federal Reserve Bank of New York staff nowcasting model uses 37 variables to nowcast US GDP \citep{FRBNY}. The need to extract useful information from high-dimensional time series data has only increased as new indicators have become available. For instance, \citet{mccracken2016fred} describe the St Louis Federal Reserve Monthly Database for Macroeconomic research, which contains 134 monthly US indicators.

The macroeconomic nowcasting literature has tackled dimensionality issues in a variety of ways. The leading approaches include dynamic factor models (DFMs), Bayesian vector auto-regressions (BVARs), and penalised estimation methods like the LASSO.\footnote{Recent macroeconomics papers have also investigated the use of other machine learning methods including random forests and artificial neural networks \citep{richardson2021nowcasting}.} 

Dynamic factor models were introduced into economics by \citet{Geweke1977} and are reviewed in \citet{StockWatson2010} and \citet{BaiNg2008}. These models are based on the idea that most of the time series variation in a large set of economic indicators is driven by the dynamics of a small number of unobserved common factors. These factors can be estimated from the complete set of economic variables using singular value decomposition of the data matrix. The dynamics of the factors are analysed using standard time series methods such as vector autoregression. There has been rapid development of DFM methods in applied macroeconomic analysis, and they have been used extensively in economic forecasting and nowcasting \citep{GiannoneReichlinSmall2006,doz2006,bok2018macroeconomic}.

Bayesian vector auto-regressions (BVARs), developed for macroeconomics by \citet{litterman1979techniques} and \citet{sims1980macroeconomics} (see \citet{karlsson2013forecasting} for a review), regularise and limit the dimensionality of models by first specifying informative prior beliefs on the parameter space and then using Bayesian methods to estimate the posterior distribution of the parameters. The posterior distribution is then used to nowcast the economic variables of interest. These methods are generally restricted to linear models, and so depend on careful choice of parameterisation.
They also require care to ensure that the priors capture relevant information -- for example momentum effects and non-stationarity typically need to be explicitly included.

Penalised estimation methods limit the dimensionality of models by including a penalty term in the estimation objective function. For example, the LASSO estimator, popularised in statistics by \citet{tibshirani1996regression}, adds a penalty based on the sum of the absolute values of the parameters (known as the L1 penalty). This shrinks parameter estimates towards zero and, due to the non-smooth penalty function, selects only a subset of variables to have non-zero coefficients. The ridge estimator \citep{hoerl1970ridge} instead adds a penalty based on the sum of squares of the parameters (known as the L2 penalty). This also has the effect of shrinking parameter estimates towards zero. The Elastic Net \citep{zou2005regularization} combines the L1 and L2 penalties. Many other penalty functions are possible. Penalised estimation allows researchers to include very large numbers of predictors in nowcasting models and, in many settings, improves the results of nowcasts by reducing over-fitting. These methods have been used in nowcasting in \citet{babii2022machine}, for example.

In Section~\ref{sec:NYfedapplication}, we show how to combine dynamic factor models with a signature method to produce nowcasts for US GDP growth. Unfortunately, we cannot directly compete with the NY Fed Nowcast \citep{bok2018macroeconomic} given the restrictions on the set of variables publicly available. However, on the same dataset, our Signature-DFM model improves nowcast accuracy by about 20\%.  

%% file: theory.tex
\section{Theory}\label{sec:background}

In this section, we outline the underlying mathematical theory for nowcasting with state-space models and illustrate their connection to the path signature. We start with a recap of the discrete time Kalman filter and its continuous time extension, then introduce the theory of signatures, and finally show how regression on signatures subsumes the Kalman filter. We introduce here notation that we use throughout the paper.

\subsection{Discrete time state-space models}

First, we recap the classic discrete-time Kalman filter (see \citet{bertsekas2012dynamic} for an overview).

Suppose that we have a hidden process $\{ Y_t \}$ with $Y_t \in \mathbb{R}^d$ that we cannot directly measure, but would like to infer from an observed process $\{ X_t \}$ where $X_t \in \mathbb{R}^m$.
In the context of nowcasting, $Y_t$ would be e.g. current GDP (which is not observed due to publication lags), and $X_t$ is a vector of observable variables at time $t$.
At each time $t$, the Kalman filter is deployed in two stages. The first is the \textit{prediction} stage: given a previous estimate of $Y_{t-1}$, predict $Y_t$. The second is the \textit{correction} stage: as more information about $X_t$ arrives, update the estimate of $Y_t$.

To be precise, assume that the ground truth is given by the linear equations:
$$Y_t = AY_{t-1} +  W_t,\qquad X_t = CY_t + V_t,$$
with initial condition $Y_0 \sim N(\mu_{0|0}, P_{0|0})$.
Here $W$, $V$ are independent white noise processes in $\mathbb{R}^{d}$ and $\mathbb{R}^m$ respectively, with $W_t\sim N(0, \Gamma)$ and $V_t \sim N(0, \Sigma)$, mutually independent at all time points. Note $A$ and $\Gamma$ are in $\mathbb{R}^{d\times d}$, $C\in \mathbb{R}^{m\times d}$, and $\Sigma\in \mathbb{R}^{m\times m}$. Assume all parameters are known and write \(\mathcal{X}_t = (X_1, \dots, X_t)\) to describe observations up to time \(t\).

For the prediction stage, $Y_t|\mathcal{X}_s$, $t \ge s$, is normally distributed, so we can write $Y_t|\mathcal{X}_s \sim N(\mu_{t|s}, P_{t|s})$.
Using the dynamics of $Y$ and $X$, we obtain the prediction equations:
\begin{equation*}
    \begin{split}
    \mu_{t|t-1} &\equiv  E[Y_t|\mathcal{X}_{t-1}] =A \mu_{t-1|t-1},\\
    P_{t|t-1} & \equiv \mathrm{var}(Y_t|\mathcal{X}_{t-1}) = A P_{t-1|t-1} A^\top +  \Gamma.
    \end{split}
\end{equation*}

Given new information, 
the updated mean estimate is
\begin{equation}
\label{eq:DiscreteKalmanMean}
\mu_{t|t} \equiv E[Y_t|\mathcal{X}_{t}] = E[Y_t|X_t,\mathcal{X}_{t-1}]= \mu_{t|t-1} + (P_{t|t-1} C^\top S_{t}^{-1}) (X_t - C \mu_{t|t-1}),
\end{equation}
while the variance correction equation is given by the discrete-time Riccati equation
\begin{align}
\label{eq:DiscreteKalmanVariance}
P_{t|t} \equiv \mathrm{var}(Y_t|X_t, \mathcal{X}_{t-1})&= (I- (P_{t|t-1} C^\top S_{t}^{-1}) C)  P_{t|t-1}. 
\end{align}
This method generalises to the case where the parameters are time-varying.

\subsection{Continuous time state-space models}\label{sec:cts_kf}

We wish to work with irregular and mixed frequencies. In order to treat timings of observations that are variable, it is natural to embed the discrete-time model in a continuous-time framework. As pointed out in \citet{harvey1990forecasting}, p.479, \textit{``[a]lthough missing observations can be handled by a discrete time model, irregularly spaced observations cannot. Formulating the model in continuous time provides the solution. Furthermore, even if the observations are at regular intervals, a continuous time model has the attraction of not being tied to the time
interval at which the observations happen to be made.''} (Also see \citet{bergstrom1984hoe}.) The most natural continuous-time extension of the Kalman state-space model is the Kalman--Bucy filter.\footnote{For a detailed study of these equations, and derivation of the filter, see for example \citet{bain2008fundamentals} or \citet[Chapter 22]{cohen2015stochastic}.} As in the discrete time setting, assume a hidden process $Y$ (this is the economic variable we are interested in, e.g. current GDP) and observed process $X$. In this case, assume the ground truth is of the form:
\begin{align*}
    dY_t &= (FY_t+f)dt + \sigma dV_t,\\
    dX_t &= (H Y_t + h)dt + \tilde{\sigma} dW_t,
\end{align*}
where $F\in \mathbb{R}^{d\times d},\: \sigma \in \mathbb{R}^{d\times p},\: f\in \mathbb{R}^d,\: H\in \mathbb{R}^{m\times d},\: h\in \mathbb{R}^m,\: \tilde\sigma \in \mathbb{R}^{m\times \tilde p}$. Hence, $Y$ is $d$-dimensional and $X$ is $m$-dimensional. (We can also allow the coefficients to be deterministic functions of time, with minimal changes.) The processes $V$ and $W$ are Brownian motions and (for notational simplicity) we assume $W$ is independent of $V$. By computing the Moore--Penrose pseudoinverse $\tilde{\sigma}^\dagger$, and hence replacing $X$ with $\tilde{\sigma}^\dagger X$ (and modifying $H$ and $h$ accordingly), we can assume without loss of generality that $\tilde{\sigma}$ is an identity matrix, for ease of presentation.

The filter, which estimates the current state of the underlying process $Y$, is then characterised by the following pair of stochastic differential equations
\begin{align}
    d\hat Y_t &= (F \hat Y_t +f)dt + R_tH^\top (dX_t -(H\hat Y_t+h)dt),\label{eqn:opt_filter}\\
    \frac{dR_t}{dt} &= \sigma \sigma^\top+F R_t + R_t F^\top - R_t H^\top H R_t. \label{eqn:filter_var}
\end{align}
Here, \(\hat{Y}_t = \mathbb{E}[Y_t \mid \mathcal{F}^X_t]\) denotes the optimal estimate (in the mean square sense) of the unobserved state \(Y_t\), conditioned on the information available from the observed process \(X\) up to time \(t\), that is, \(\mathcal{F}^X_t = \sigma(X_s : s \leq t)\).
The matrix \(R_t \in \mathbb{R}^{d \times d}\) is the conditional covariance of the estimation error
\[
R_t = \mathbb{E}[(Y_t - \hat{Y}_t)(Y_t - \hat{Y}_t)^\top \mid \mathcal{F}^X_t].
\]
Note that \(R_t\) evolves deterministically according to the Riccati Equation \eqref{eqn:filter_var}, so it does not depend on the realisation of the observed process \(X\), but only on the system parameters and time. Once \(R_t\) is solved, it can be used in Equation \eqref{eqn:opt_filter} to obtain the filter estimate \(\hat{Y}_t\) from the observed data. The pair \((\hat{Y}_t, R_t)\) fully characterises the conditional distribution of \(Y_t\) given the observations: \(Y_t \mid \mathcal{F}^X_t \sim \mathcal{N}(\hat{Y}_t, R_t)\). In essence, up to changes of notation, the continuous-time mean equation \eqref{eqn:opt_filter} is completely analogous to the discrete-time mean equation \eqref{eq:DiscreteKalmanMean}, and similarly for the variance equations \eqref{eqn:filter_var} and \eqref{eq:DiscreteKalmanVariance}. Unsurprisingly, the discrete-time equations naturally arise as (Euler--Maruyama) discretisations of the continuous dynamics on a regular grid. However, a key advantage of the continuous-time perspective is that irregular observations can also be considered, leading to different discretisations.

It is worth noting some characteristics of the (discrete or continuous) Kalman filter. First, the various parameters are assumed known -- typically in practice these must be learned from data.  (This can be done, for example, using the EM algorithm \citep{dempster1977maximum}.) More significantly, the system is assumed to follow linear dynamics. This can be relaxed using variants such as the Extended Kalman Filter, which assumes the processes are nonlinear stochastic dynamic processes \citep{jazwinski1970stochastic}. 

\subsection{Paths and signatures}\label{sec:sigs}

The discussion above has highlighted a tension between modelling and the complexities of real-world data: in order to make dynamic estimates, many traditional models are based on linear (Kalman-like) state space models and stochastic filtering, which can be considered in discrete or continuous time; at the same time, our observations are often irregular and asynchronous, and the models we may wish to consider can be nonlinear, for which the Kalman filtering theory does not apply. This latter challenge leads us to look for other techniques that include the Kalman filter as a special case, but allow for more flexible models and observations. 

If we had regular and synchronous observations, a na{\"\i}ve approach would be to simply replace filtering with nonlinear regression on our observations. For example, we could build an estimate of $y_t$ based on  a large polynomial on all past observations of $x$. A key difficulty of this is that, when $x\in\mathbb{R}^k$, a polynomial of degree $m$ has ${k+m}\choose{m}$ terms, which grows very rapidly in $k$ and $m$. If our observations are irregular, we additionally need to record the time of each observation, which increases the number of terms further. Even with this additional time dimension, difficulties remain when the number of observations available varies over time. Consequently, a simple nonlinear regression approach has not been particularly popular as a method for nowcasting.

The approach we will propose is based on the insight that we can treat a wide variety of nowcasting problems, with irregular and mixed frequency observations, by widening our view to consider our observations as `paths' in time and space. By thinking of paths, rather than discretely observed time-series, as our fundamental object, we obtain a unified approach, and open up the use of a variety of approximations based on the well-established mathematics of paths. We then define an appropriate notion of a `polynomial of a path', based on the path signature representation, and use this as the basis for regression.\footnote{We emphasise that this is not the same as a polynomial in a collection of individual observations, even when the path is based on interpolating finitely many observations. This is because the transformation from observations to paths and then to signatures encodes information in a different way, allowing much better representation properties in settings where observations are irregular or at high frequency.} This enables us to implement a version of the nonlinear regression approach to nowcasting, and show that this gives good performance for practical problems.

In the rest of this section, we introduce terminology and definitions related to paths and signatures. For a more detailed introduction see \citet{chevyrev_primer_2025, fermanian2021embedding, lyons2007differential}.

Formally, a path in \(\mathbb{R}^d\) is defined as a continuous function \(X: [a,b] \to \mathbb{R}^d\), where each component is a one-dimensional path \(X^{k}: [a, b] \to \mathbb{R}\), $k \in \{1, \dots, d\}$. For a given sequence of observations, there are a variety of ways we can think of recovering a `path' from these observations. Suppose the $j$th component is observed at times $\{t_{i, j}\}_{i\in\mathcal{I}_j}$, for some countable index $\mathcal{I}_j$. A path can be constructed by first considering the map $t\mapsto (t, \{X^j_{\lfloor t\rfloor_j}\}_{j})$, where $\lfloor t\rfloor_j = \max_i\{t_{i,j}:t_{i,j}\leq t\}$ gives the time of the most recent observation of the component $X^j$. This function is not continuous, but has the simple interpretation that it represents the `most recent observations' (or equivalently, is the result of forward filling\footnote{As the goal of this construction is the path, the choice of forward filling does not imply that we believe unobserved values are well approximated by constants, but simply acts as a representation of our observation values.} our observations). We can then extend this function by connecting the jumps in observations, to obtain the `rectilinear path interpolation'. This is illustrated in Figure \ref{fig:rectilinear}.

\begin{figure}
    \centering
    \caption{Rectilinear interpolation. Here real observations (red) are first forward filled, resulting in a piecewise constant function. In order to obtain a continuous path (without changing the flow of information), the values immediately before (blue) and after (red) a new observation are connected vertically using `virtual time'. The resulting rectilinear interpolated path has a longer length than the original observations, but describes the available information precisely in a consistent manner for all sampling regimes.}
    \includegraphics[width=0.7\linewidth]{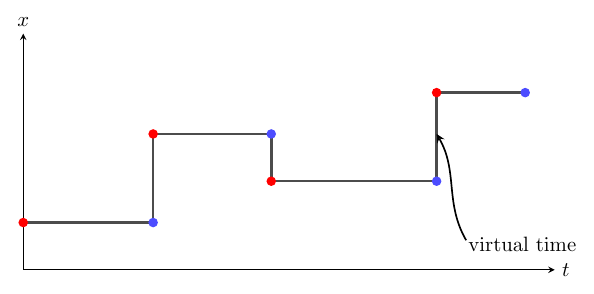}
    \label{fig:rectilinear}
\end{figure}

This construction of a path is somewhat artificial, but demonstrates that any sequence of discrete observations can be represented by a path, without changing the flow of information through time. Importantly, this allows us to represent all observation sequences (at mixed frequencies and with irregular observations) in the same space, as all of them can be considered to be paths. In particular, we do not have to assume that our observations should be sampled at high-frequency, as the `observation path' is a well-defined object even with infrequent observations\footnote{One might ask whether there is a sense that, when observations become increasingly frequent, the path should converge to a continuous observation path. This is indeed the case, as we discuss in Appendix~\ref{app:consistency}. In particular, the infill asymptotic limit is studied in \cite{lucchese2025learning}, and a convergence rate is given in terms of the frequency of observations. However, in many problems, observation frequencies are not within our control, and so we will think of the path constructed from discrete observations as the primary object, rather than as an approximation of an unobtainable ideal path.}.

The difficulty of working with paths is that it is not immediately apparent how to use these as statistical objects. The path signature is a mathematical object that describes time-series-like data and is closely related to the theory of rough paths (which deals with the integration theory for nowhere-differentiable functions, and the related controlled ordinary differential equations, see \citet{friz2020course}). The path signature is a property of continuous paths and has been shown to be an effective description of intertemporal information in prediction tasks -- for example in Chinese handwriting recognition \citep{graham2013sparse}, sepsis detection \citep{morrill2020utilization}, malware detection \citep{cochrane2021sk-tree}, and time series generation \citep{ni2021sig-wasserstein}. Intuitively, the signature allows us to describe `polynomial' functions of paths, in a mathematically rigorous way.

The signature is defined through an iterative construction. Let
    \begin{equation*}
        S^1(X^k)_{a,t} \equiv \int_{a}^t dX_s^{k} = X_t^{k} - X_a^{k}, \quad k \in \{1, \dots, d\},
    \end{equation*}
     and let the double iterated integral be
    \begin{equation*}
        S^2(X^{k}X^l)_{a,t} \equiv \int_a^t S^1(X^k)_{a,s} \: dX_s^{l} = \int_a^t\int_a^s dX_r^{k}dX_s^{l}, \quad k, \: l \in \{1, \dots, d\}.
    \end{equation*}
More generally define the n-fold iterated integral by
\begin{equation*}
    S^n(X^{k_1}X^{k_2}\dots X^{k_{n-1}} X^{k_n})_{a,t} \equiv \int_a^t S^{n-1}(X^{k_1}X^{k_2}\dots X^{k_{n-1}})_{a,s} \: dX_s^{k_n}, \quad k_i \in \{1, \dots, d\}.
\end{equation*}
    
The superscript in \(S^2(X^{k}X^l)_{a,t}\) denotes the number of iterated integrals, the subscript indicates the limits of the outermost integral, and the terms inside the brackets give the order of integration, starting with the innermost integral. The number of iterated integrals taken is the ``level'' of the signature. The signature of the path is the ordered infinite collection of all such iterated integrals
\begin{equation*}
    S(X)_{a,b} \equiv  \Big(1, S^1(X^1)_{a,b}, \dots,S^1(X^d)_{a,b}, S^2(X^1X^1)_{a,b}, S^2(X^1X^2)_{a,b}, \dots \Big).
\end{equation*}

The path signature captures geometric information. The first level signature terms give the change/increment in each dimension between the start and end of the path. The second level terms are linked to areas bounded by the path. Figure~\ref{fig:sig_illustration} offers a visual representation of the intervals and areas captured by the first and second level signature terms for a two-dimensional path $X: [0,T] \rightarrow \mathbb{R}^2$ starting at the origin, with the dimensions denoted by $x$ and $t$. Notice that $S^1(x)_{0,s} = x_s - 0 =x_s$ and $S^1(t)_{0,s}=t_s-0=t_s$, i.e. the increments in each dimension. For second level terms, we have $S^2(x,t)_{0,s}=\int_0^{s} S^1(x)_{0,u} \: dt_u = \int_0^s x_u \: dt_u$, $S^2(x,x)_{0,s}=\int_0^{s} S^1(x)_{0,u} \: dx_u = \int_0^s x_u \: dx_u = 1/2 x_s^2$, and similarly $S^2(t,x)_{0,s}=\int_0^s t_u \: dx_u$ and $S^2(t,t)_{0,s} = 1/2t_s^2$.  These correspond to the areas highlighted in Figure~\ref{fig:sig_illustration}. In particular, cross terms, $S^2(X^iX^j)_{a,b}, \: i\neq j$, register how some dimensions tend to change before others (as the relative sizes of $S^2(tx)$ and $S^2(xt)$ in Figure~\ref{fig:sig_illustration} will change) and are related to the L\'{e}vy area for the path.\footnote{For a two-dimensional path, the L\'{e}vy area $A$ is given by
\begin{equation*}
    A = \frac{1}{2} \big(S^2(X^1X^2)_{a,b} - S^2(X^2X^1)_{a,b}\big).
\end{equation*}
The L\'{e}vy area encodes the order of events,
where a positive value of $A$ indicates $X^1$ is typically followed by $X^2$. See \citet{Yang2017DevelopingTP} for further discussion.} These and higher signature levels thus serve as a ``geometric footprint'' for the path.

\begin{figure}
    \caption{Illustration of the first two levels of the signature for a two-dimensional path (with variables \(t\) and \(x\)).}
    \centering
    \includegraphics[width=0.6\linewidth]{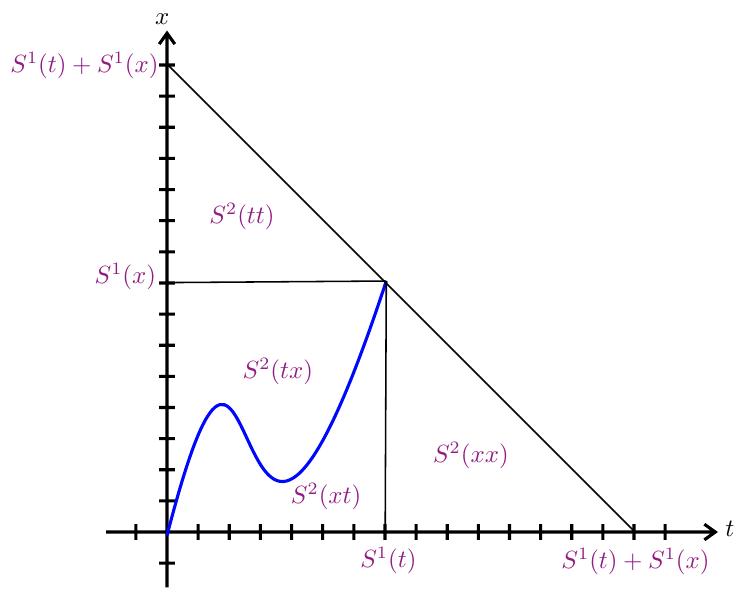}
    \label{fig:sig_illustration}
\end{figure}

Considering the signature as a class of approximating functions, there is a (Stone--Weierstrass) theorem that gives universality of this approximation class \citep{levin2013learning}.
That is, any continuous function of paths can be approximated arbitrarily well through a linear combination of signature terms.\footnote{For further information on approximation results, see Appendix~\ref{app:consistency}.} Intuitively, as mentioned above, linear functions of signatures behave like ``polynomial functions of paths''. Because of this, all estimation procedures involving learning from paths can be approximated (arbitrarily well) using a linear function of the signature.

As signature terms are iterated path integrals, they inherit the invariance properties of path integrals. Namely, the terms are invariant to translation of the path and to any monotonic increasing reparameterisation in time.

The mathematical theory underpinning signature methods is well-developed, with several key results supporting both theoretical analysis and practical computation. For example, Chen's identity \citep{chen1954iterated} provides a way to compute the signature of a concatenated path\footnote{A path formed by joining two or more path segments end-to-end.} as an algebraic product of individual path segments.
This identity allows for recursive or piecewise computation of signatures, making it especially useful for when the data arrives in batches, as long paths can be processed incrementally rather than recomputed from scratch. As a result, Chen's identity plays a vital role in the efficient and scalable computation of signatures.

\subsection{Practical signature computation}\label{sec:practicalsig}

In the previous section, we explained that the signature is a natural object for describing paths in general and time-series data in particular. However, when working with real data, we only observe variables at discrete intervals, rather than continuous paths. As a result, we need to interpret new data releases as coming from some underlying continuous time process and exploit connections to model the target variable. Using signatures as covariates is a natural way to make sense of data released irregularly in continuous time and sidesteps the mixed frequency and missingness challenges covered in Section \ref{sec:nowcasting_challenges}. 

As signatures describe continuous paths, but data are typically discrete, we must form a path from these observations, e.g. through interpolation, before computing the signature. There are many interpolation methods, and choosing between these is a modelling choice that also depends on the points in time when we want to make a prediction. The choice of interpolation can matter, and has been explored \citep{morrill2022onthe, walker2026faithful}. Computation of the signature depends on the paths constructed from the discrete samples rather than directly on the frequency of sampling. Smooth cubic splines are suitable when we wish to make a single inference using the maximum amount of information available (irrespective of the time that they were observed). Rectilinear interpolation (as discussed above) is more appropriate for data measured at different times and frequencies when a prediction may be needed at any point. In nowcasting, we have an online/causal problem (that is, we only want predictions to be made based on available information) where we wish to make inferences at given points: in this case, the interpolation method employed should be either continuously online (rectilinear interpolation) or discretely online (e.g. forward filling) (see \cite{morrill2022onthe} for more details).

The number of signature terms at each level is dependent on the dimension of the path: this means that if the number of explanatory variables is fixed, the number of signatures \texttt{at each level} is fixed too.
In practice, the iterated integrals that define signatures can be efficiently computed with software packages such as \texttt{iisignature} \citep{iisignature} and \texttt{esig} in Python.

It is nevertheless important to note that the number of terms increases exponentially with the level of the signature, that is, for a $d$-dimensional path, the signature truncated at the $k$th level, contains \(\sum_k d^k\) terms. So, while signatures can universally approximate any path objects, they do so with an exponentially increasing number of terms, and including all of these would increase model complexity and could result in over-parameterisation. 

However, an important result in rough paths limits the extent to which the exponentially increasing number of terms by level is a problem: the magnitude of signature terms decays factorially with increasing levels \citep{lyons2007differential}. If we take \(X:[0,T] \to \mathbb{R}^d\) to be a path with finite one-variation (total variation), i.e. 
\[||X||_{1,[0,T]} \equiv \sup_{P \in \mathcal{P}} \sum_{i=0}^{n_P-1} |X_{i+1} - X_i| <\infty,\]
where \(\mathcal{P}\) is the set of all partitions of \([0,T]\) and \(n_P\) is the number of points in partition \(P\), then, for each level \(k \geq 1\), 
\begin{equation}\label{eqn:sig_magnitude}
    S^k(X\dots X)_{[0,T]} = \int_0^{T}\cdots\int_0^{t_{1}} dX_{t_0}  \dots dX_{t_{k-1}} \leq \frac{||X||^k_{1,[0,T]}}{k!}.
\end{equation}

Therefore, as the level of signature terms increases, the magnitude of the terms decreases factorially -- which is faster than exponentially -- and we still obtain a good approximation to the signature by truncating it at a sufficiently high level.

There are two different ways to incorporate information into nowcasting models, and train prediction protocols in general: \textit{expanding window}, where new information is appended to existing data, or \textit{rolling window}, in which data older than a certain lookback period is discarded from the model. When using an expanding window, certain terms in the signature may be increasing. For example, the first-level signatures terms are just the differences/increments between the start and the end of the path, and therefore the signatures corresponding to time can increase as more observations become available.

Due to the invariance to reparameterisation property of signatures, it is common to manually add time as one of the input variables if the rate at which things occur is relevant (as is typical in prediction problems). Adding time as one of the variables is referred to as \textit{time-augmentation}. This also ensures that signatures are unique \citep{hambly2010uniqueness}. In the methods we consider, we will always make use of the time-augmented path. Other information that can boost performance, especially in contexts where data are not missing at random, include the time since the last observation or counts of the number of observed measurements \citep{morrill2022onthe}.

\subsection{The signature as a source of covariates in prediction}\label{sec:sigprediction}

In this section, we briefly review the use of signatures as covariates and highlight some variations of ``signature methods'' in other applications.

One very useful property of truncated signatures is that if two paths are ``close'' (under the $1-$variation norm), then the truncated signatures of these paths will also be close \citep{giusti2020iterated}.
Combined with the universal approximation property (by the Stone--Weierstrass Theorem \citep{levin2013learning}), signatures therefore should be good candidates to generate covariates for prediction or classification tasks. Further, because linear combinations of signatures can effectively capture nonlinear relationships, even linear regression via ordinary least squares is theoretically sufficient to model nonlinear problems.

Regression on signature covariates has been used in high-frequency financial time series \citep{lyons2014afeature, levin2013learning} and in freight transportation price forecasting, where it provided an annualised saving of \$50m for Amazon \citep{gu2024transportation}. \citet{fermanian2020linear} looked at functional linear regression where information on predictors over time is available; this is similar to the setting of economic nowcasting. \citet{fermanian2020linear} showed that signature regression is competitive with traditional functional regression methods such as functional principal components and B-splines.

The truncation level can be chosen in different ways. One of the simplest is hyper-parameter search \citep{morrill2021neuralrough} but others have taken more systematic approaches. For example, \citet{fermanian2020linear} assumes that the true model is linear in some truncated signature space. An estimator for the truncation level is then designed to minimise the sum of the mean square error of the model and a penalty based on model complexity. This estimator is used with Ridge regression for their functional regression model. \citet{bleistein2023learning} show that, using LASSO for regularisation, the truncation error decays exponentially fast with the truncation level.

Note that there are other choices to be made in deriving signature covariates, some of which might be called ``model tuning'', that we do not fix \textit{a priori}. Some attempts have been made to systematically categorise the types of operations and model choices that can be made. For example, \citet{morrill2021generalised} split these choices into what they call ``modifications'' with four primary groups: augmentations, windows, transforms, and rescalings.
We forgo further discussion of these techniques, in the interest of focusing on the core simplicity of the method.

\subsection{Signatures generalise linear Kalman filters}\label{sec:signKalman}
In this section we show that the Kalman--Bucy filter of Equation (\ref{eqn:opt_filter}) is equivalent to a linear regression problem on the signature space. In other words, we can express $\hat Y$ as a linear combination of the iterated integrals of the time-augmented observation process $(t, X_t)$. Though this can be derived more abstractly from the general universal properties of the signature previously mentioned, we provide a formal proof.

The equivalence result we prove holds in any dimension of observation and filter processes. Therefore, it includes the full range of (continuous time) vector autoregressive and linear state-space models. This suggests that linear regression on signatures has the benefit of avoiding many of the modelling restrictions needed in order to apply filtering methods in situations that are of great practical relevance.

\begin{theorem}\label{thm:KalmanRepresentation}
The continuous time Kalman--Bucy filter $\hat Y$ (including the case with deterministic time-varying parameters), as described by Equation \eqref{eqn:opt_filter}, can be written as a linear function of the initial estimate $\hat Y_0$ and the signature of the augmented observation process $(t,X_t)$. In particular, this representation uses only the signature terms where the observation $X$ appears once in the iterated integral.
\end{theorem}

\begin{proof} With the natural time-varying change to the notation set out in Section~\ref{sec:cts_kf}, write
\begin{align}
A_t &\equiv F_t -R_t H_t^\top H_t,\nonumber\\
\xi_t &\equiv  \int_0^t (f_s - R_s H_s^\top h_s) ds + \int_0^t R_sH_s^\top dX_s, \label{eqn:xi}
 \end{align}
which gives us the simplified expression for the filter
\begin{equation}\label{eq:simplifiedXhat} d\hat Y_t = A_t \hat Y_tdt + d\xi_t.\end{equation}

Denote the iterated integrals of $A$ and $\xi$ by
\begin{align*}
    \mathbb{A}^n_t &= \int_0^t\int_0^{t_1}\cdots\int_0^{t_{n-1}}\Big(A_{t_1}A_{t_2}\cdots A_{t_n}\Big)dt_{n}\cdots dt_1,\\
    \Xi^n_t &= \int_0^t\int_0^{t_1}\cdots\int_0^{t_{n-1}}\Big(A_{t_1}A_{t_2}\cdots A_{t_n}\xi_{t_n}\Big) dt_{n}\cdots dt_1,
\end{align*}
with the conventions $\mathbb{A}^0 =\mathrm{I}_d$ (the identity matrix) and $\Xi^0_t = \xi_t$. Observe that these are both formally solutions of the recurrence relation $Q^n_t = \int_0^t A_s Q^{n-1}_s ds$ with different values for $Q^0$ (and in different dimensions). We note that the $n$th iterated integral is super-polynomially small (that is, it satisfies Equation (\ref{eqn:sig_magnitude})), and in particular the infinite sums $\sum_{n\ge 0} \mathbb{A}^n_t$ and $\sum_{n\ge 0} \Xi^n_t$ are both well defined.

Writing $\hat Y$ in integral form, we have
\[\hat Y_t = \hat Y_0 +\int_0^t A_s \hat Y_s ds + \xi_t.\]
This integral equation has a unique solution, which can be represented as the series
\[\hat Y_t = \sum_{n\ge 0} \Big(\mathbb{A}^n_t \hat Y_0 + \Xi^n_t\Big).\]
Given the recurrence relation satisfied by $\mathbb{A}^n_t$ and $\Xi^n_t$, we can see that this is a solution to the integral equation.
This shows that $\hat Y_t$ is a (linear) function of its initial value $\hat Y_0$ and the iterated integral processes $\mathbb{A}^n$ and $\Xi^n$. This expansion is very closely related to the Picard series approximation of the stochastic differential equation defining $\hat Y$ (Equation \eqref{eqn:opt_filter}). 

It remains to show that $\mathbb{A}^n$ and $\Xi^n$ can be expressed in terms of the (joint) signature of the time-augmented path $(t, X_t)$. As $A$ is a continuous deterministic function of time, we know that (over any finite time horizon) it can be approximated arbitrarily well by a polynomial (by the Stone--Weierstrass theorem). As the signature of time is simply the sequence $\mathbb{S}(t)= (1, t, t^2/2, \dots, t^n/n!,\dots)$, we see that $A$ can be written as a (matrix-valued) linear function of the signature of $t$.

The process $\xi$ in Equation~(\ref{eqn:xi}) is slightly more delicate, as it depends on both time and the observations $X$. The first term $\int_0^t (f_s - R_s H_s^\top h_s) ds$ is deterministic, so again has a polynomial expansion in terms of the signature of $t$. Considering the second term $\int_0^t R_sH_s^\top dX_s$, we see that if $H$ is continuous\footnote{If $H$ is not continuous, then we need to use the fact that polynomials in time are dense in the $L^2([0,t])$ space, which is the appropriate space to consider given we have an outer integral with respect to the process $X$. In this case we will still have a polynomial approximation $p_n(s) \approx R_sH_s^\top$, such that the integral $\int_0^t p_n(s) dX_s \to \int_0^t R_s H_s^\top dX_s$ converges in a mean-square sense as $n\to \infty$.}, $R_sH_s^\top$ has an expansion in terms of the signature of time, so $\int_0^t R_sH_s^\top dX_s$ has an expansion in terms of the signature of $(t,X_t)$, with the special form where the only integral with respect to $X$ is the outermost one.

As both $A$ and $\xi$ have expansions in terms of the signatures of $t$ and $(t,X)$ respectively, it follows that their iterated integrals $\mathbb{A}$ and $\Xi$ also have expansions of this type. For our purposes the explicit values of this expansion are not of particular interest (and do not have a simple algebraic form), but the definition of $\mathbb{A}^n$ immediately shows that if $A$ can be written as a polynomial in time, then so can $\mathbb{A}^n$ for each $n$. Similarly for $\Xi^n$, but now this will involve iterated integrals with a single $X$ integral included.

\end{proof}

The absence of iterated integrals with more than one $X$ integral in the theorem above is equivalent to the linear dependence of the filter $\hat Y$ on the observations $X$. We refer to the iterated integrals with only a single $X$ integral as ``linear'' signature terms. 

The construction above shows that we can expand $\hat Y$ in terms of the initial estimate $\hat Y_0$. The next lemma shows that we have great flexibility when constructing our signature expansion for the filter, as we can use the signature over any sufficiently large horizon in our representation of the filter.
\begin{lemma}\label{lemma:longerwindow}
For a continuous path $\hat Y$, and any time $t\in [0,T]$, the value of $\hat Y_0$ can be expressed as a linear function of $\hat Y_t$ and the linear signature terms of $(t, \hat Y)$. Consequently, we can express the Kalman--Bucy filter $\hat Y_t$ as a linear function of $\hat Y_s$ and the linear signature of $X$ on $[s',T]$, for any $s'\leq s < t$.
\end{lemma}
\begin{proof}
We first focus on the expression for $\hat Y_0$. The iterated integrals of time give polynomials (of the form $t^n/n!$). By using the inner products
\begin{align*}
   \langle \hat Y, t^n/n!\rangle_{L^2} &= \frac{1}{n!}\int_0^T \hat Y_s s^n ds\\
   &= \hat Y_0 \frac{T^n}{n!} + \frac{1}{n!}\int_0^T \Big(\int_0^s d\hat Y_u\Big)s^n ds
\end{align*}
and similarly $\langle t^m/m!, t^n/n!\rangle$.
A simple application of integration by parts shows that the final integral term of the above equation can be expressed as the sum of two signature terms, of the form $S^n(\hat yt\cdots t)$ and $S^n(t\cdots t\hat y)$.

Suppose we compute a regression expansion of \(Y_t\) in terms of \(\{t^m/m! \}_{m \leq M}\).
As $\hat Y$ is a continuous path, and the polynomials are dense in the $L^2$ space on $[0,T]$, 
this (theoretical) expansion perfectly approximates \(Y_t\) as we take \(M \to \infty\).
It follows that $\hat Y_t$ can be evaluated by evaluating a linear function of the signature terms and $\hat Y_0$. Rearrangement gives the desired expression for $\hat Y_0$ in terms of $\hat Y_t$.

As we can write $\hat Y_0$ in terms of the signature, we now focus on the representation of $\hat Y_t$ in terms of $\hat Y_s$ and the signature of $X$ on $[s', T]$. By translation, we can assume $s'=0$. We have seen that $\hat Y_0$ can be expressed in terms of $\hat Y_s$ and the linear signature of $Y$. From Theorem~\ref{thm:KalmanRepresentation}, we know that $\hat Y$ can be expressed as a linear function of $\hat Y_0$ and the linear signature of $X$. Therefore, by a direct calculation, the linear signature terms of $\hat{Y}$ can also be expressed as a linear function of $\hat Y_s$ and the linear signature of $X$, as required.
\end{proof}

We conclude with an example where we explicitly compute the optimal filter in terms of the signature.

\begin{example}\label{ex:derive_coeffs_cts}
Consider the situation where $d=p=m=1$, $f=h=0$, $F=-1$, $H=1$ and $\sigma=\sqrt{3}$. Then our filter equations simplify to 
\begin{align*}
    d\hat Y_t &= -\hat Y_tdt + R_t(dX_t -\hat Y_tdt),\\
    \frac{dR_t}{dt} &= 3-2 R_t - R_t^2.
\end{align*}
For simplicity, assume the initial variance is at the steady state $R_0=R_t = 1$, so we have
\[d\hat Y_t = -2\hat Y_t dt + dX_t,\]
which can be solved as 
\begin{align*}
    \hat Y_t &= e^{-2t}\hat Y_0 + \int_{0}^t e^{-2(t-s)} dX_s.
\end{align*}
Assuming \(X_0=0\), the signature expansion of $\hat Y$ can also be computed, as was done for Equation~\eqref{eq:simplifiedXhat}, with the identity $A_t= -2$ and $\xi_t=X_t$, to give\footnote{Here we write $S^n(\cdots)$ for the $n$-fold iterated integral with respect to the sequence indicated (which must be of length $n$, with the innermost integral listed first); i.e. $S^3(xtt) = \int_0^t\int_0^{t_1}\int_0^{t_2} 1\,dX_{t_3}\, dt_2\, dt_1$. Note that in this example all signature terms are taken over the interval $[0,t]$, so we drop the subscript on signatures for convenience.}
\begin{align*}
    \mathbb{A}^n_t &= (-2)^n \frac{t^n}{n!} = (-2)^n S^n(t\cdots t)\\
    \Xi^n_t &= (-2)^n S^{n+1}(xtt\cdots t),
\end{align*}
and hence Theorem \ref{thm:KalmanRepresentation} yields
\begin{align*}
    \hat Y_t &= \sum_n \Big( \mathbb{A}^n_t \hat Y_0 + \Xi^n_t\Big)\\
    &= \hat Y_0+\sum_{n\ge 1} \Big([(-2)^n\hat Y_0] S^n(t\cdots t) + [(-2)^{n-1}] S^{n}(xtt\cdots t)\Big)\\
    &= \hat Y_0\Big(1 -2  t + 2 t^2 - \frac{8}{6}  t^3 + \frac{16}{24} t^4+\dots\Big)\\
    &\qquad + X_t - 2\int_0^t X_s \,ds + 4 \int_0^t\int_0^{t_1} X_s \,ds\, dt_1 \\
    &\qquad -8\int_0^t\int_0^{t_1}\int_0^{t_2} X_s\, ds\, dt_1\, dt_2+\dots
\end{align*}
In particular, if $t$ is small, the first few terms of this series provide a good approximation for the value of $\hat Y_t$, in terms of a linear function of the signature. Note that $\hat Y_t \approx \hat Y_0 -2\hat Y_0 t + X_t$, the single-step (Euler--Maruyama) approximation \citep{kloeden1992numerical} of Equation~\eqref{eqn:opt_filter}, appears as the first terms in this expansion.

The simple structure we obtain here is due to the assumptions we have made on our state-space model (in particular, the absence of signature terms of the form $S^n(t\cdots txt\cdots t)$ is due to the assumption the variance is in its steady state).
\end{example}

The key advantage of this approach, even while restricting our attention to the Kalman--Bucy state-space model, is that we now have an expansion which is valid\footnote{In practice, as with polynomial approximations of the exponential function, the number of terms needed for large $t$ can increase rapidly, along with their coefficients, leading to numerical instability. For this reason, the performance of the approximation is likely to degrade as the time between official releases (where the underlying value $Y$ is observed) increases.} for all $t$. This simplifies dramatically the problem of working with data at mixed frequencies, as we can evaluate the filter state at any $t$, in terms of the corresponding signature terms, rather than having to compute (as is done, for example, in a MIDAS model) a version of the filter that depends on the timing of observations.

A further advantage of signature methods is that they allow us to easily incorporate nonlinearity into our state-space models. For example, if our observations $X$ were replaced in the Kalman--Bucy setting by $\bar X = \exp(X)$, then, by It\^{o}'s chain rule, we know that
\[d\bar X_t = \bar{X}_t d X_t + \frac{1}{2} \bar{X}_t dt 
\iff dX_t = \frac{1}{\bar{X}_t}d\bar{X}_t - \frac{1}{2}dt.\]
Substituting into Equation~(\ref{eqn:opt_filter}), we see that the Kalman--Bucy filter can be written in terms of the modified observation process $\bar X$, and by very similar arguments to before, this would have an expression in terms of the signature, but with terms involving multiple integrals with respect to $\bar X$.

This suggests that using the signature expansion is robust to the specification of the observation time series. That is, \textit{the usual issues around the choice of transformations (whether to use a series, its logarithm, differences, etc...) are less significant when using the signature method, as they are absorbed into the signature expansion}. The price that we pay to capture these nonlinear relationships is that the expansion contains multiple integrals with respect to $\bar X$, so we lose the ability to write the filter with only the ``linear'' signature terms and hence all signature terms should be retained in model fitting.

Finally, in Appendix \ref{app:consistency}, we explore the connection between regression on signatures and the consistency properties of ordinary least squares. We show that the overall error of our method can be decomposed into three different factors, depending on the horizon chosen, the signature expansion of the estimator, and the regression error based on observed data; all three of these can be made arbitrarily small (at least asymptotically when we have a large amount of data), demonstrating the overall consistency of our framework.

In summary, we have shown that signatures generalise a wide range of nowcasting methods and their variants in any number of dimensions and even with nonlinearities, and that regression on signatures has theoretical properties that make it appropriate for use in time series problems.

%% file: sig_nowcast.tex
\section{Nowcasting via regression on signatures}\label{sec:sigregression}

In this section we present how to use regression on the signatures of observations for nowcasting.

\subsection{Model specification}

Let $Y_{t}$ be the (low-frequency) target variable we wish to nowcast. Let $X_{t}$ denote the (high-frequency) observed explanatory variables at \(t\). 
We wish to build an estimate of $Y_t$ based on a previous observation \(Y_{t-}\) and observations of \(\{ X_s \}_{s \in [s', t]}\), where \(s'\) is the start of the lookback window.
Motivated by the approximation of the Kalman filter in terms of signatures, we suppose that:
\begin{equation}
Y_t = \sum_{k=0}^{\infty} \left(\alpha_k+\beta_k Y_{t-} \right) \psi_{k,t} + \epsilon_t,  
\label{eqn:regressionsig}
\end{equation}
where
\begin{itemize}
    \item $Y_{t-}$ is a previous observation of the low frequency target variable that is available at the time of nowcast. This could be the target variable at the beginning of the lookback window over which the signatures are computed, or the most recent available observation;
    \item $\psi_{k,t}$ is a sequence (for each value of $t$) of signature terms at level \(k\), including iterated integrals of $t$ and the different components of the observed process $X$, calculated over the lookback window $[s',t]$ ending at the present;
    \item $\epsilon_t$ is a stationary error term with $E\big[ \epsilon_t \big| Y_{t-}, \{\psi_{k,t}\}_{k=0}^\infty \big]=0$;
    \item $\alpha_k, \: \beta_k$ are sequences of regression coefficients.
\end{itemize}
As outlined in Section~\ref{sec:signKalman}, in order to replicate the Kalman filter, we would have $\beta_k=0$ whenever $\psi_k$ corresponds to a signature term depending on $X$, and \(\alpha_k, \:\beta_k\) nonzero only for those signature terms that depend either purely on $t$, or on a single component of $X$ appearing once in the iterated integral. Both of these restrictions can be relaxed, leading to a richer class of models than can be captured by the (linear) Kalman filter.

In practice, for a fixed finite sample, we must project onto a finite approximation by truncating the terms of the signature at some level $K<\infty$. As discussed in Section~\ref{sec:sigprediction}, there are a variety of practices for choosing the signature truncation level \(K\). We follow \cite{morrill2021neuralrough} and select the level through hyper-parameter optimisation. For real-world data, we search over a range of signature levels (in covariates and time) up to level $4$: this captures both pairwise geometric relationships between series and higher-order interactions among multiple variables (see Section~\ref{sec:sigs}).

Note that if the gaps \(t-t_-\) vary significantly over the prediction period, it would make sense to account for this by using a richer class of models, where \(t-t_-\) gives another covariate for the signature computation.
Note also that the signature terms are typically highly multi-collinear. This can be addressed through any standard technique, e.g. using Principal Component Analysis (PCA) to capture the most important information, or using only a smaller number of terms in the expansion.

In the previous section, we explained that the signature is a natural object for describing paths. Moreover, signatures can be efficiently computed with software packages such as \texttt{iisignature} \citep{iisignature} and \texttt{esig} in Python. However, when working with a finite sample of data, there are a number of practical issues that must be addressed before computing the signatures. First, we must decide on the observation window used to construct the signature. Second, to approximate the integrals defining the signature, we must use some interpolation method to fill in missing information about the path. Third, we must decide on the truncation level: how many levels of signature terms to retain in the model.

For the observation window, we choose to use a rolling window of fixed length $T_w$. That is, for each time $t$, data older than $t-T_w$ (the lookback period) is discarded from the analysis. Just as in non-parametric regression, increasing $T_w$ reduces variance (when the processes are stationary), but introduces bias when there are time-varying parameters. The constructed signature terms also depend on the lookback period length since, for example, the first level signature terms are increments between the start and the end of the path. The lookback period length $T_w$ is a hyper-parameter that we choose as part of our hyper-parameter search. Detailed algorithms for regression on signatures can be found in Appendix~\ref{sec:framework}.

The path is a continuous time process, but data are discretely observed. The discrete points must be interpolated to compute the signature. In a finite sample, the choice of interpolation can matter, as has been explored \citep{morrill2022onthe}. In the in-fill asymptotic regime, the signature of the interpolated path converges to the signature of the continuously observed path (see Appendix~\ref{app:consistency}). In practice, we consider rectilinear interpolation and forward filling.

\subsection{Regularisation}\label{sec:regularisation}

If the number of predictors is large, the number of signature terms can still be large even if we restricted the truncation level to $K=4$. Standard regularisation techniques can be used for feature selection and to reduce over-fitting.

We use elastic net regularisation \citep{zou2005regularization}, which combines L1 penalty and L2 penalty terms, and therefore can also perform feature selection to reduce the number of signature terms used. The optimisation problem being solved is given by
\begin{align*}
    \min_{\alpha_k, \beta_k} \: \bigg\{ & \frac{1}{n}\bigg|\bigg|Y_t - \sum_{k=0}^K (\alpha_k+\beta_kY_{t-}) \psi_{k,t}\bigg|\bigg|_2^2 \\
    & +  \sum_{k=0}^K \bigg(\gamma \lambda (||\alpha_k||_1 + ||\beta_k||_1)
+ \gamma (1 - \lambda) (||\alpha_k||_2^2 + ||\beta_k||^2_2 \bigg) \bigg\},
\end{align*}
where \(n\) is the number of time points that data are collected from,  \(\gamma\) is the regularisation strength parameter, and \(\lambda\in (0,1)\) is the L1 ratio. 
The parameters of the elastic net, \(\lambda\) and \(\gamma\), are found through hyper-parameter optimisation.

Recall that the magnitude of signature terms decreases at a factorial rate with the truncation level, $K$ (see Equation \ref{eqn:sig_magnitude}). Because elastic-net regularisation is sensitive to the scale of the regressors, applying the penalty directly to unnormalised signature terms would disproportionately penalise higher-level terms due to their relatively smaller magnitude. Therefore we standardise the signatures by removing the mean and scaling to unit variance once truncated terms have been obtained (but before using them in regression). For this, we use the \texttt{StandardScaler} from \texttt{sklearn} \citep{pedregosa2011scikit}.

\subsection{Rolling windows for signatures}

In Section~\ref{sec:signKalman} we saw that we can replicate the Kalman filter with a linear combination of signatures, where the linear signature terms are multiplied by the true value of the hidden process at some past time. In the context of nowcasting, to account for non-stationary behaviour of time series, we can compute signatures of our observations over a rolling or sliding window as discussed in  Section~\ref{sec:practicalsig}. 

Given that we wish to maximise the use of alternative data sources that are more timely than the target variable, our window of information should include at least the last release of the target variable. We therefore assume that the ``true'' value of the target variable will be available to the model at some point in the window that we compute signatures in, and that we can use that true value as a multiplier. This is denoted by \(Y_{t^-}\) in Equation \eqref{eqn:regressionsig}. 

Alternatively, a previous estimate can be used. As explained in Lemma \ref{lemma:longerwindow}, the multiplier can be at any fixed time within the lookback window -- but note that different past times will result in different regression coefficients.

\subsection{Hyper-parameter selection}
Regression on signatures forms a family of approximation functions, with a range of hyper-parameters to optimise over. These include the lookback window length $T_w$, the truncation level $K$, elastic net penalty parameters \(\lambda\) and \(\gamma\).

Hyper-parameters are obtained by running the models over a validation period. The best hyper-parameters are then used to fit the final model on the whole training and validation period, and this calibrated model is used in the test period.

%% file: simulation.tex
\section{Simulation: equivalence to the Kalman filter}\label{sec:simulation}

In Section \ref{sec:signKalman}, we proved that the signature method subsumes the linear Kalman filter. 
As discussed in Appendix~\ref{sec:reg_error}, standard consistency results follow if we have access to collections of independently sampled training paths. 
We now illustrate this theoretical result in a controlled environment, by comparing the performance of the Kalman filter to regression on signatures (as described in Section \ref{sec:sigregression}) on simulated data, and building intuition before analysing real-world applications. Specifically, we consider four cases: linear observed process with regularly sampled data, linear observed process with irregularly sampled data, nonlinear observed process with regularly sampled data, and a nonlinear observed process with irregularly sampled data. In each case, we compare the results from the signature method to those from the theoretically optimal Kalman filter (with perfect information of the system). We find that the residuals of the signature method are highly correlated with those from the Kalman filter in all settings except the final case of nonlinear process with irregularly sampled data. However, even in this last case, the mean and the variance of the residuals from the signature method remain competitive. This study on simulated data demonstrates that signatures can effectively handle both nonlinearity and irregular sampling at essentially no extra effort. In the subsequent sections, the signature method is applied to US GDP data (where there are regular monthly updates on all the predictors) and UK unemployment data (where the data is mixed frequency) -- these provide complementary and more challenging case studies for our method. 

\subsection{Simulation specification}

We simulate a $1-$dimensional continuous-time state-space model similar to Example~\ref{ex:derive_coeffs_cts}. The latent state \(Y_t\) and its associated linear observation process \(X_t\) follow the SDEs
\begin{align*}
    dY_t &= -Y_t dt + \sqrt{2} dV_t, \\
    dX_t &= 10Y_t dt + dW_t.
\end{align*}

Throughout this section, we consider two observation regimes. In the first case, we observe \(X_t\) directly (where the Kalman filter is optimal). In the second, we observe $\bar X_t$, which is a nonlinear, sigmoid transform of \(X_t\): 
\[
    \bar X_t = g(X_t) = \frac{1}{1+e^{-X_t}}.
\]
In this setting, the na\"{i}ve Kalman applied to \(\bar{X}_t\) directly is not optimal, the baseline filter is obtained from \eqref{eqn:opt_filter} by replacing $dX_t$ with $\bar g(\bar X_t)\,d\bar X_t$ (i.e. working with the transformed increments), where
\[
    \bar g(\bar X_t) := \frac{d g^{-1}}{d\bar X_t} = \frac{1}{\bar X_t(1-\bar X_t)}.
\]

In Example~\ref{ex:derive_coeffs_cts}, by assuming that the variance is in steady state, we derived the coefficients of the optimal filter, \(\hat{Y}\), with respect to its infinite signature expansion. 
In particular, under the linearity assumptions, the only non-trivial terms in the expansion are signatures that we refer to as ``linear'' (where only one of the iterated integrals is with respect to the observed data and the rest are time integrals), and further, the only integral with respect to the data variables is the innermost one. We adopt these assumptions for the linear observation regime. In contrast, under the nonlinear observation regime we can no longer accurately represent the filter using only the innermost linear signatures; instead we must retain all signature terms up to a specified truncation level\footnote{This highlights the potential flexibility of the signature approach, when compared with traditional linear filtering techniques (such as the Kalman filter or MIDAS). As the signature regression can directly incorporate nonlinearities (by including more signature terms as regressors), it can be applied in settings where a linear assumption may be inappropriate. For this reason, we can expect that it will typically perform better than linear methods, at least when trained with sufficient quantities of data and levels of regularisation.}.

In both cases, we assume the simulated path starts at \(Y_0 = 0.1\). Assuming steady-state variance, we can compute \(R=H^{-2}\big(F+\sqrt{F^2+\sigma^2H^2}\big)\). The end time, \(T\), is randomly sampled uniformly on \([0.1,1]\) and the mesh size in time is fixed at \(\Delta t = 0.005\).

For the linear observed process, because the dynamics are linear and the true parameter values are known, the ideal (generally infeasible) Kalman filter minimises mean square error. The mean of the filter is given by Equation (\ref{eqn:opt_filter}). The discretised optimal filter is then given by 
\[\hat{Y}_{t+\Delta t} = \hat{Y}_t + (F-RH^2)\hat{Y}_t\Delta t+RH(X_{t+\Delta t}-X_t).\]

To compare regression on signatures against this optimal but infeasible filter with perfect knowledge of the parameter values, we proceed as detailed in Section \ref{sec:sigregression}. The purpose of this exercise is to verify that a regression on signatures as in (\ref{eqn:regressionsig}) (even if truncated at a finite signature level, with finite observation window and possibly irregularly sampled) has the capacity to approximate the optimal filter in light of the results in Section~\ref{sec:signKalman}. Unlike the Kalman filter then, we utilise a large training set of independent sample paths to estimate the (truncated) signature regression coefficients.\footnote{See Appendix~\ref{app:consistency} for discussion of consistency in both the independent and dependent sample cases.}

We thus simulate $1000$ independent paths (for both \(X_t\) and \(Y_t\)) of which we take $800$ to be in a training set (to calibrate the regression model) and the remaining $200$ in a test/evaluation set.
For this simulation, we use an expanding window to compute the signatures.

We compute signatures of each observed path (\(X_t\) in the linear regime and \(\bar{X}_t\) in the nonlinear regime). In the linear regime, we truncate at signature level $6$, keeping only the linear signatures of the form \(S^n(xtt\cdots t)\) as discussed (Example~\ref{ex:derive_coeffs_cts}).
For the nonlinear regime we use \textbf{all} signatures up to signature level $3$ in the regression, yielding a comparable number of covariates ($15$ versus $13$).

Overall we conduct four simulation experiments: linear process with regularly sampled data, linear process with irregularly sampled data, nonlinear process with regularly sampled data, and a nonlinear process with irregularly sampled data. To determine whether the Kalman filter and regression on signatures are similar, we compare the means and variances of their respective residuals.

For the irregularly sampled cases, we randomly drop 80\% of the regularly sampled data from the example above; this is done uniformly at random. The signature-based model is recalibrated using this reduced dataset.

\subsection{Simulation results}

For the linear case with regular sampling, Figure~\ref{fig:sim_residuals_composite} shows the final residuals (difference between target and estimated values at the final point on the 200 test paths) of the Kalman filter and the signature method plotted against one another with a line-of-best-fit. The two approaches perform similarly on the test set. The alignment of the errors of the signature method (where the parameters are fitted from data) against the ideal filter (given full information) is excellent. The means and variances are similar, suggesting that these methods are equivalent means of reaching the same results. The signature method achieves this result without knowledge of the true data generating process.

\begin{figure}[ht]
\centering
    \caption{Residuals of the Kalman filter vs the signature method on regularly sampled data for the linear process. Marginal plots are histograms of residuals. Line-of-best-fit shown in red (gradient $1.00$ and \(y\)-intercept of $-0.02$).}    \includegraphics[width=0.75\linewidth]{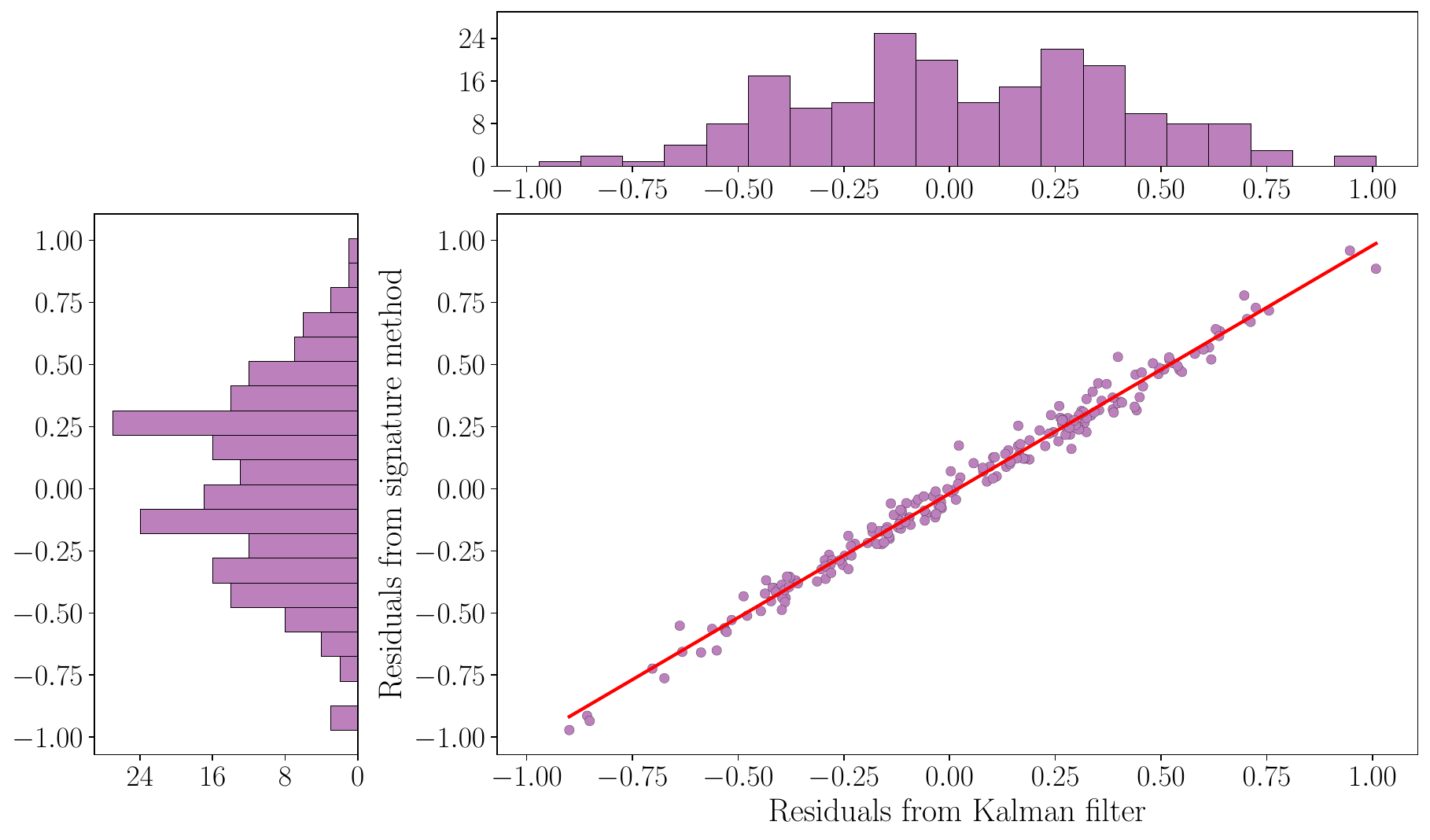}
  \label{fig:sim_residuals_composite}
\end{figure}

To give a sense of how these models perform in nowcasting contexts, we apply them both recursively on new observations to see how well they infer the path of a hidden target variable. We illustrate the recursively inferred paths on a single example from the test set in Figure~\ref{fig:traj}. The signature method and Kalman filter produce similar trajectories despite the former facing the additional difficulty of inferring unknown parameters from data. Furthermore, both inferred paths are close to the true underlying path, \(Y_t\).

\begin{figure}[ht]
    \centering
    \caption{A regularly sampled simulated path (from the test set) with observed values (\(X_t\)) along with the true, hidden values that need to be inferred (\(Y_t\)).}    \includegraphics[width=0.75\textwidth]{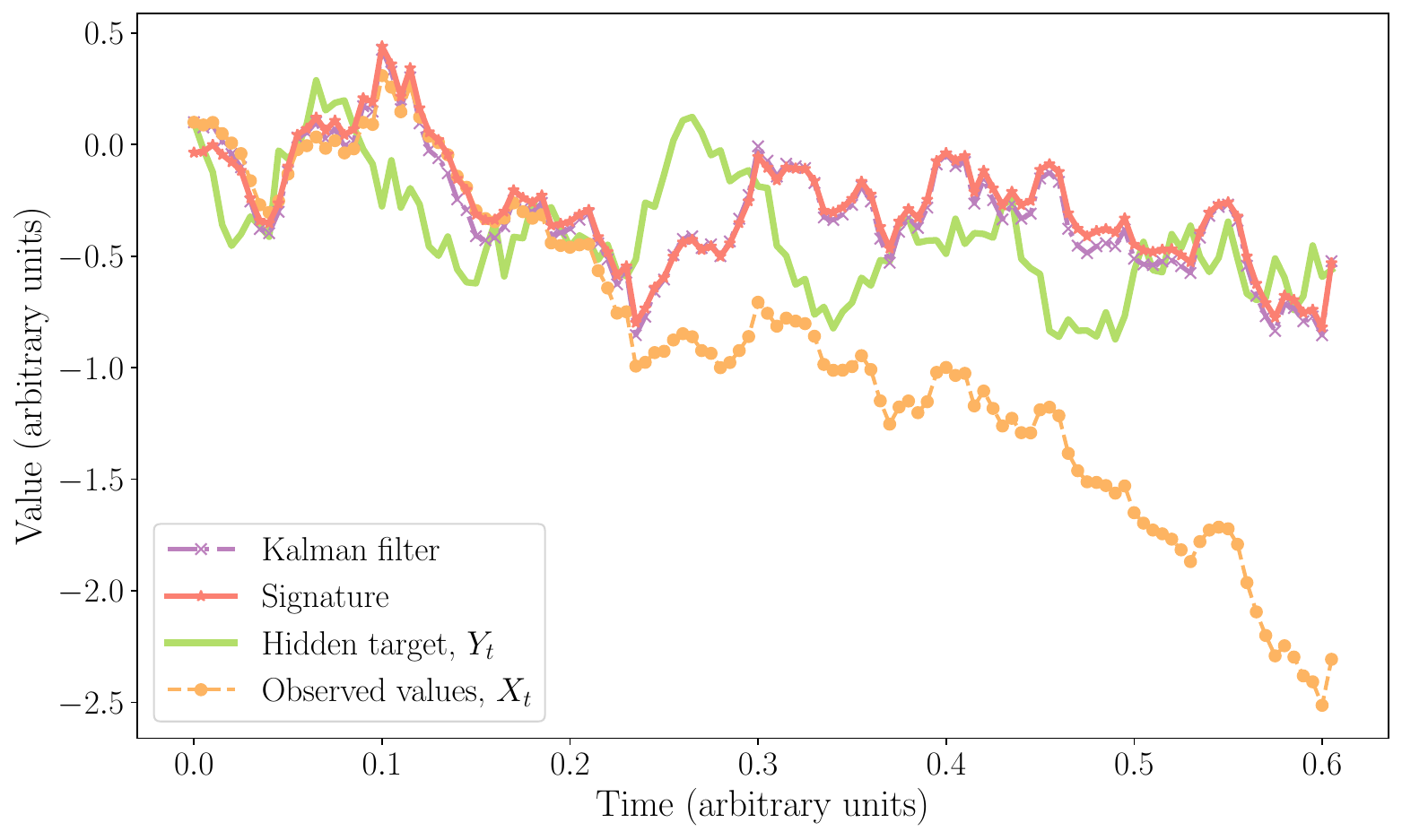}
    \label{fig:traj}
\end{figure}

\begin{table}[h!]
\caption{Means and variances from simulations (test set): Kalman filter and signature method.}
\label{tab:simsummary}
\begin{tabular}{lrrrr}
\toprule
& \multicolumn{2}{c}{Linear} & \multicolumn{2}{c}{Nonlinear} \\
\cmidrule(lr){2-3}\cmidrule(lr){4-5}
Name & Regular & Irregular & Regular & Irregular \\
\midrule
Kalman residuals (mean) & 0.04 & 0.05 & 0.02 & 0.01 \\
Signature residuals (mean) & 0.02 & 0.02 & 0.03 & 0.03 \\
Kalman residuals (variance) & 0.14 & 0.15 & 0.14 & 0.64 \\
Signature residuals (variance) & 0.14 & 0.14 & 0.24 & 0.24 \\
\bottomrule
\end{tabular}
\end{table}

For each of the other three simulation experiments, charts similar to Figures~\ref{fig:sim_residuals_composite} and \ref{fig:traj} are presented in Appendix~\ref{app:simulation}. Here, we summarise the results in Table~\ref{tab:simsummary}. The variances of the residuals are extremely similar for the linear cases, and not considerably worse for the nonlinear cases. The relationship is harder to learn in the nonlinear case, especially in the irregularly sampled case, when we drop 80\% of the regularly sampled data at random. However, despite the observed data having a nonlinear relationship with the true, hidden values \(Y_t\), a regression on signature terms obtains predictions that closely match the theoretical optimum. In summary, a regression on signatures can pick out the key information even when the data are more sparse and nonlinear.

%% file: us_gdp.tex
\section{Nowcasting US GDP growth}\label{sec:NYfedapplication}
In this section we apply our proposed methodology to the real-world task of nowcasting US GDP growth. Specifically, we compare our regression on signatures model to the version of the Dynamic Factor Model (DFM) developed by \citet{bok2018macroeconomic}, which was employed by the New York Fed for their staff nowcast prior to the post COVID-19 update by \citet{almuzara2023new}. 

We refer to the model proposed by \cite{bok2018macroeconomic} as the ``NYFed model". It provided weekly nowcasts of quarterly US GDP quarter-on-quarter growth using a DFM estimated with 37 monthly macroeconomic variables, released at varying dates over the quarter. In this section we show that nowcasting with signatures (as defined in Section~\ref{sec:sigregression}) combined with standard dimension reduction methods (e.g. principal components) outperforms the NYFed model. Unfortunately, we cannot use the exact dataset as proposed by \citet{bok2018macroeconomic} due to data availability constraints (see Section~\ref{US_now_data} for details). Therefore, we compare our model with two benchmarks: (i) the published NYFed staff nowcasts and (ii) the same DFM structure using only publicly available data. This is a standard, well-studied nowcasting problem. Because we consider only variables of the same (monthly) frequency, without any missing data, such a setting does not fully exploit the key strengths of signatures. Nonetheless, our proposed methodology still outperforms both benchmarks, reducing nowcast error by up to 21\%.

\subsection{Data}\label{US_now_data}

In order to nowcast US GDP, we use $33$ of the $37$ variables listed in \citet{bok2018macroeconomic}. The four variables we omit (``ISM non-manufacturing: NMI composite index'', ``ISM mfg.: Prices index'', ``ISM mfg.: PMI composite index'', and ``ISM mfg.: Employment index'') are not publicly available. We use University of Michigan: Consumer Sentiment \citep{curtin2008consumer} as a replacement for the ISM series, meaning that we use $34$ variables in total. The complete dataset includes a selection of monthly variables that covers housing, income, manufacturing, the labour market, manufacturing activity surveys\footnote{\href{https://fred.stlouisfed.org/series/GACDISA066MSFRBNY}{Empire State Mfg Survey: General business conditions} and \href{https://fred.stlouisfed.org/series/GACDFSA066MSFRBPHI}{Philly Fed Mfg. business outlook: current activity}.}
, trade, and consumption. The indicators are monthly, but due to varying publication lags and frequencies, new information is released each week.

For further details see Appendix~\ref{sec:usgdp_factors} and Table~5.1 of \cite{bok2018macroeconomic}.

\subsection{US GDP nowcasting models}

\subsubsection*{Dynamic Factor Model (DFM)}

Our baseline model is a DFM based on the NYFed model. The model assumes that the \(k\) observed variables, \(X_t \in \mathbb{R}^k\), are driven by $r$ unobserved factors ${F_t} = (F_{1,t}, \dots, F_{r,t})^\intercal$. In state space form, the model is
\begin{align*}
    X_t &= \Lambda F_t + e_t, \\
    F_t &= \sum_{i = 1}^{p} A_i F_{t-i} + V_t, \\
    e_t &= \sum_{i = 1}^{q} B_i e_{t-i} + W_t,
\end{align*}
where \(\Lambda \in \mathbb{R}^{k \times r}\) represents the factor loading, \(e_t \in \mathbb{R}^k\) is the idiosyncratic noise process in the observations, and \(V_t \in \mathbb{R}^r\) and \(W_t \in \mathbb{R}^k\) are independent Gaussian random noises. The parameters \(p\) and \(q\) are the orders of the autoregressive processes of the factors and the idiosyncratic noise process respectively. The NYFed model assumes  \(p=q=1\).

In factor models, the factor loadings and factors are not identified without imposing further structure. We follow \citet{bok2018macroeconomic} and impose a dedicated four-factor structure on the factor loading matrix with the following economic interpretation. The first factor is a ``soft'' factor, which only includes variables derived from survey data, usually published earlier than official statistic variables. The second and third factors are called ``real'' and ``labour'', which refer respectively to activity and labour market variables. The final factor is a global factor, which includes all predictor variables. Details of the variables in each group are listed in Appendix \ref{sec:usgdp_factors}. These restrictions are sufficient to identify the model parameters.

\citet{bok2018macroeconomic} extract the factors from the data with the following procedure: At the start, the factors are initialised by finding the principal components of each factor group. Then, there is a two step procedure: first, given the estimated factors, we estimate the loading matrix, dynamic coefficients, and variance parameters by least-squares regressions; second, given these estimated parameters, we update the common factors using a Kalman smoother and the Expectation-Maximisation algorithm (see \citet{baum1966statistical, baum1967aninequality, baum1970amaximization}). Iterations of this EM procedure converge to a local optimum of the likelihood function under standard regularity conditions \citep{dempster1977maximum, wu1983on}.

\subsubsection*{Regression on signatures with dimension reduction}
Here we outline modifications of the signature method to combine it with dimension reduction.

The process is repeated each time new information is received. First, principal component analysis (PCA) is used to find the first principal component of each factor group. We enforce the same dedicated factor structure as in the NYFed dynamic factor model. Then we compute signatures on time and on the time series of factors derived from PCA. Finally, we run a regularised regression of US GDP on the signatures and on the latest available GDP value.\footnote{Note that this is the latest \textit{published} value -- for US GDP, there is a 30-day lag between the end of a quarter and the publication of the data.} Effectively, this means the model learns the change in GDP growth. We label this procedure Signature-PCA.

In fact, given that the signature method is largely agnostic to the data used as inputs, we can also use it as a method to improve existing nowcasts. For example, we can take the mean estimated factors from the NYFed DFM, compute the signature of these ``data'', and then use these new signatures to nowcast US GDP growth. This hybrid method allows us to combine forecast methods. We refer to this hybrid method as Signature-DFM.

\subsubsection*{Model comparison}

For each model, when new information arrives (i.e. weekly), we use information from previous quarters to fit the model and then apply the model to nowcast the current quarter's US GDP growth. We use data from 1\textsuperscript{st} Jan 2000 to 31\textsuperscript{st} December 2015 as the training set, which means that they are always used to fit the regression parameters. 
From 2016 onward, each week, new data is incorporated into the model (and new factors and principal components are computed) and another nowcast prediction is made for GDP in the current quarter.

The period from 1\textsuperscript{st} January 2016 to 31\textsuperscript{st} December 2017 is used as a validation set to select hyperparameters such as the level of signatures to use and the length of the lookback window. The period 1\textsuperscript{st} January 2018 to 31\textsuperscript{st} December 2019 is then used as a ``test set'' with fixed hyperparameters, but the regression parameters are still tuned with new data releases over the test period. The list of selected hyperparameters can be found in Appendix~\ref{sec:hyperparams_us}.

We use root mean squared error (RMSE) as our prediction quality metric.

\subsection{Results}

We first present NYFed model published results, results from a DFM using only publicly available data, and results from Signature-PCA (signature methods combined with PCA). The headline result is that, in the test period, Signature-PCA outperforms the other approaches.
This can be seen in Figure~\ref{fig:usgdp_comparison_pca_test}.

\begin{figure}[ht]
    \centering
    \caption{US GDP results for the test period when the signature method is applied to the principal components of the factor groups.}
    \includegraphics[width=0.9\textwidth]{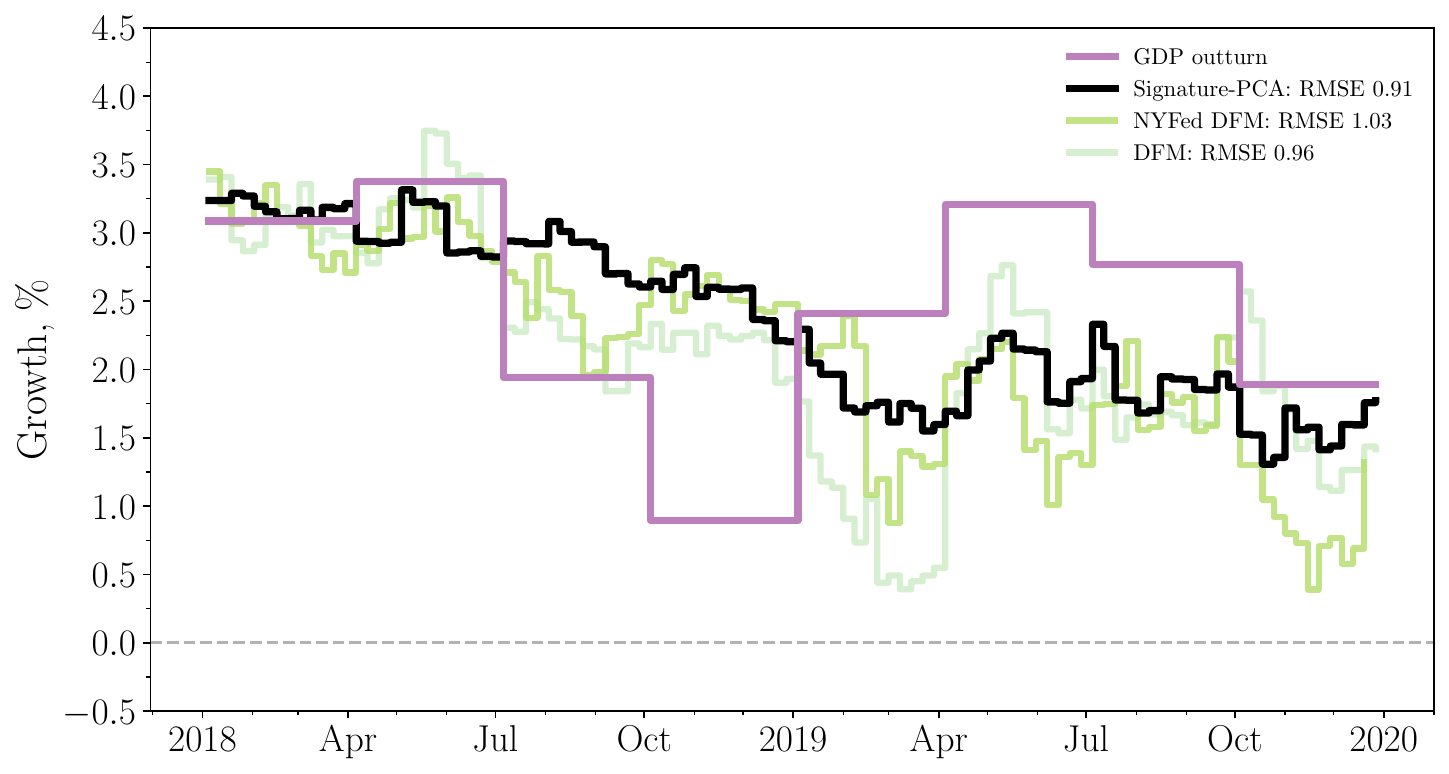}
    \label{fig:usgdp_comparison_pca_test}
\end{figure}

We also fit a hybrid Signature-DFM model. For the validation period, the Signature-DFM approach also strongly outperforms the other methods. For the test period, Figure~\ref{fig:usgdp_comparison_filtered_test} shows that Signature-DFM is superior to the DFM and to regression on signatures using PCA. Using the DFM model estimates as inputs to the signature approach produces a noticeable improvement in prediction quality over the DFM alone. This provides evidence that signatures can be incorporated into existing nowcasting models to enhance performance.

\begin{figure}[ht]
    \centering
    \caption{US GDP results for the test period when the signature method is applied to the DFM factors.}\includegraphics[width=0.9\textwidth]{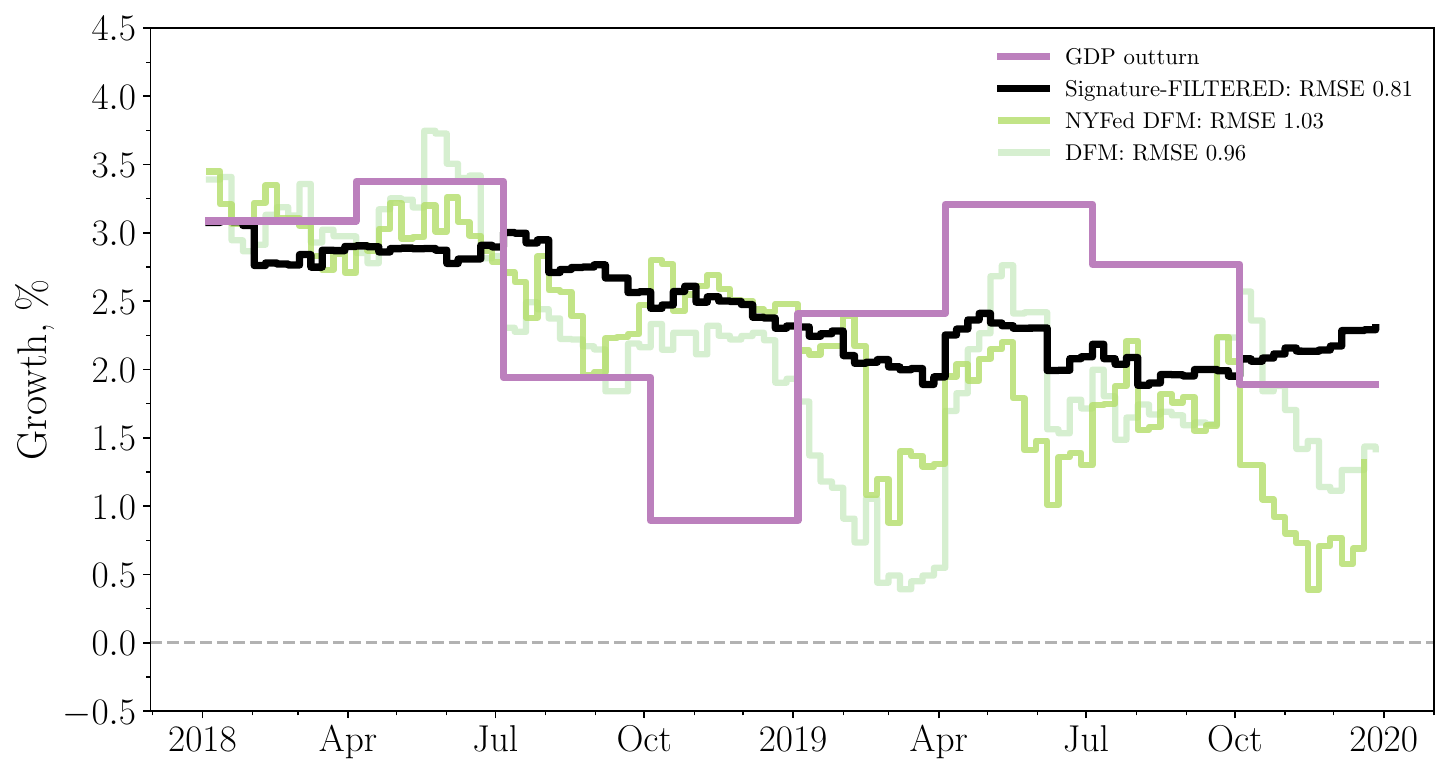}
    \label{fig:usgdp_comparison_filtered_test}
\end{figure}

\begin{table}
\centering
\caption{Root-mean squared errors across GDP nowcasting models.}
\include{us_gdp_table_rmses}
\label{tab:usgdp}
\end{table}

The results for all models on validation and test periods are summarised in Table~\ref{tab:usgdp}. We observe all models giving higher errors in the test period. Using the signature method achieves the lowest root mean squared error, with the filtered factor version of the signature method delivering the best overall performance.

%% file: us_gdp_table_rmses.tex
\begin{tabular}{lrrrr}
\toprule
Method & DFM & NYFed & Signature-PCA & Signature-DFM \\
Period &  &  &  &  \\
\midrule
Validation & 0.85 & 0.86 & 0.56 & 0.50 \\
Test & 0.96 & 1.03 & 0.91 & 0.81 \\
\bottomrule
\end{tabular}

%% file: lfs.tex
\section{Nowcasting unemployment: an application of signatures to irregular and mixed frequency data} \label{sec:lfsapplication} 

In this section we demonstrate two more properties of signatures: the ability to handle \textit{mixed frequency} and \textit{irregularly sampled} time series data. As our example, we nowcast the UK's Office for National Statistics (ONS) monthly unemployment rate (the percentage of those economically active 16- to 64-year-olds in the UK who are unemployed; labour market statistics code I46A). To do this, we use the following high-frequency variables: the weekly benefit claimant count from NOMIS (National Online Manpower Information Service); two monthly Google Trends search series (``unemployment benefit'' and ``seekers allowance''); and weekly ONS time series covering activity, inactivity, employment, and unemployment.\footnote{In the original Labour Force Survey of the ONS, it was possible to separate out responses relating to individual weeks during the survey period. While this was not representative, it could be weighted and aggregated to give a ``flash" estimate of changes in the labour force. As ONS put it, ``While not providing robust estimates of labour market conditions, these can help users to understand the impact of the coronavirus (COVID-19) pandemic on a week-by-week basis during the quarter." ONS considered these data to be ``highly experimental" but we use them here as an ideal example of the kind of imperfect but higher frequency data that may be useful during a crisis.
} The objective is not to create the most accurate nowcast of unemployment (this would be achieved by including many more predictor variables), but to demonstrate the ability of signature methods to tackle irregular and mixed frequency series.

For simplicity, we conduct only a pseudo real-time nowcasting exercise here. That is, we use all data available at the time of writing rather than as they would have been available in real-time. The key argument we are making is about the superior \textit{relative} performance of signatures (in challenging circumstances), and so the important comparison is against the other models -- which are also provided with pseudo real-time data.

\subsection{Data}

This nowcasting problem involves a complex, time-varying combination of mixed-frequency and irregular data releases. Three of the inputs we use are monthly: the two Google Trends data series and the claimant count variable. The experimental labour market statistics are weekly. Publication dates are irregular with respect to each other; because the number of weeks per month is not an integer, weekly data publication dates move around relative to the end of the month. Publication dates of both prediction variables and target also move around relative to the reference date.
There are also periods when official publications were severely delayed, either by problems at the statistical agency or by external factors such as COVID-19. 

We have historical information on reference and publication dates for official releases from 2016 onward. Before 2016, we impute the publication date to be the first Tuesday at least 45 days after the reference date. This is a very accurate rule of thumb for actual publication dates as 45 days is the modal gap for publications from 2016. Since we are using these imputed pre-2016 publication dates simply to demonstrate the performance of signature methods in a hypothetical nowcasting experiment, these imputations do not have any important impact on our results.

Each nowcast period is defined by a window that begins 30 days prior to the reference date and ends 90 days after it. This captures the events in the typical release cycle: the publication of the previous value, the arrival of contemporaneous information, and, finally, the publication of the true value. Within each nowcast period, we produce a new nowcast whenever new information is released. Because of the data irregularities noted earlier, each nowcast period sees data released at slightly different times relative to the reference date. Figure~\ref{fig:lfsinfoflow} visualises all releases in the data (2008--2023), and how the date of the publication falls relative to the reference date, using a histogram for each information type. The figure shows that, relative to the reference date, information arrives at inconsistent times. Only those releases that are actually used are shown; this is why there are few input series shown after 55 days, when, usually but not always, the unemployment rate has been published.

\begin{figure}[h!]
    \centering
    \caption{Each panel shows a histogram of publication dates relative to reference dates with frequencies counted across all nowcasting periods between 2008 and 2023. Only those variables (first four rows) used to inform estimates pre-publication (last row) are shown. Publication dates for both predictor variables and the target variable (the official release of unemployment) vary considerably relative to the reference date, except for Google Trends, which is more consistent. In particular, publication dates for experimental weekly data are not in phase with the official publication.  See Figure \ref{fig:lfsexample} for an example of the timing of information releases in a specific month.}
     \includegraphics[width=\textwidth]{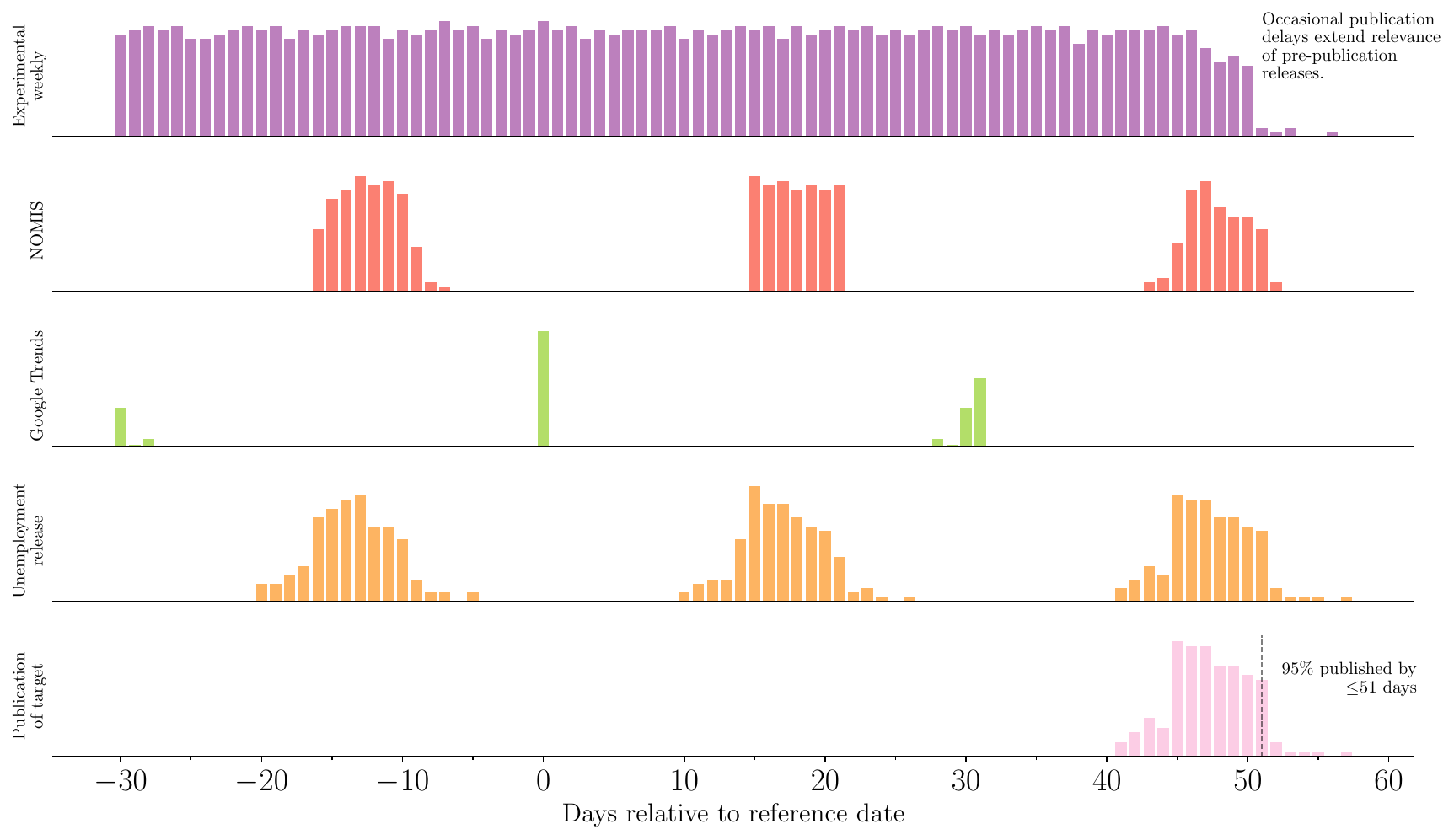}
     \label{fig:lfsinfoflow}
 \end{figure}

\subsection{Models}

Rather than directly target the official monthly unemployment rate, $u_t$, we follow \cite{Anesti2022} and nowcast the first difference.

We use the regression on signatures nowcasting model outlined in Section~\ref{sec:sigregression}, with features and data as described above. We use PCA for dimensional reduction, Elastic Net for regularisation, and we include an intercept. We only keep linear signature terms.

One of our baseline models is an auto-regressive model of order $1$, AR(1). Although simple, on average and across a wide range of situations, AR(1) models are tough to beat \citep{carriero2019comprehensive}.

As the AR(1) only uses lags of the target variable, we also compare the performance of the signature approach to the predictions of a Mixed Data Sampling (MIDAS) model as described in Section~\ref{sec:nowcasting_challenges}. We use the L-BFGS-B algorithm to solve the MIDAS model. As with the signature regression model and AR(1), we use a single lag of the target variable in the MIDAS model, as well as $4$ lags of each high-frequency (weekly) predictor (4 being roughly the number of weeks in a month).

For all models, we use the period from 2008 to 2020 as both the training and validation set, and 2020 to 2023 as the test set. This arises from the convention in machine learning of using 80\% of the data as the training and validation set. In Section~\ref{sec:lfsrobustness}, we perform robustness checks by varying the size of the training and validation set according to \citet{rossi_outofsample_2012} and show that the relative performance gain of signatures is statistically significant across a wide range of training and validation set sizes.

\subsection{Results}

After a hyperparameter search, the selected configuration has $\alpha=0$, meaning no regularisation and, effectively, OLS -- this is not unexpected given the small number of features. A truncation level of $4$ is selected and no standardisation of the features is used. The full details of the selected hyperparameters can be found in Appendix~\ref{app:lfshyper}.

\begin{figure}[h!]
    \centering
    \caption{An example nowcast of UK unemployment for the reference date of November 2019. The published value is shown by a red cross.}     \includegraphics[width=\textwidth]{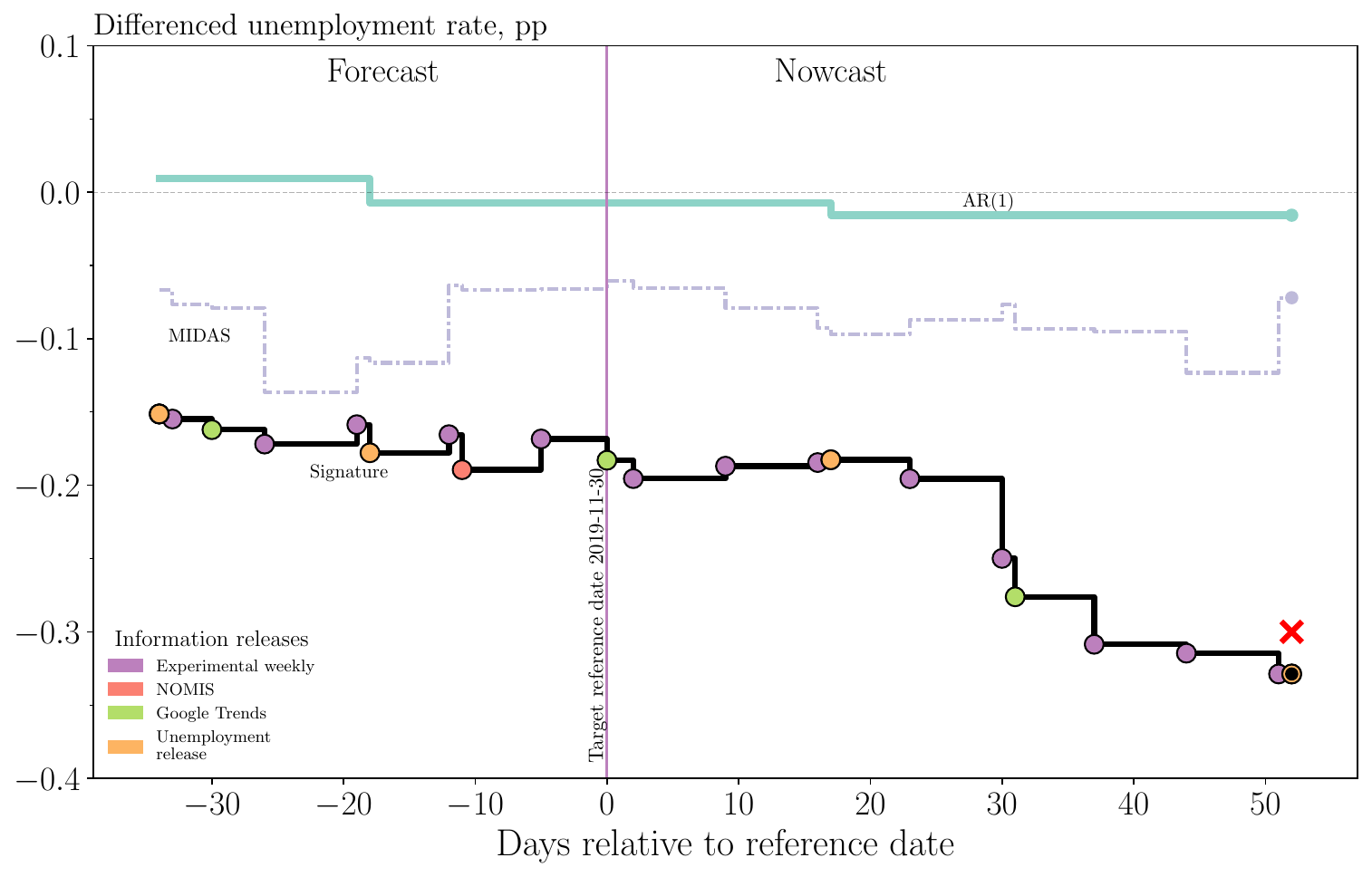}

    \label{fig:lfsexample}
 \end{figure}

\begin{figure}[h!]
    \centering
    \caption{An example nowcast of UK unemployment for the reference date of October 2021, a month after the end of the UK government's furlough scheme. The published value is shown by a red cross.}     \includegraphics[width=\textwidth]{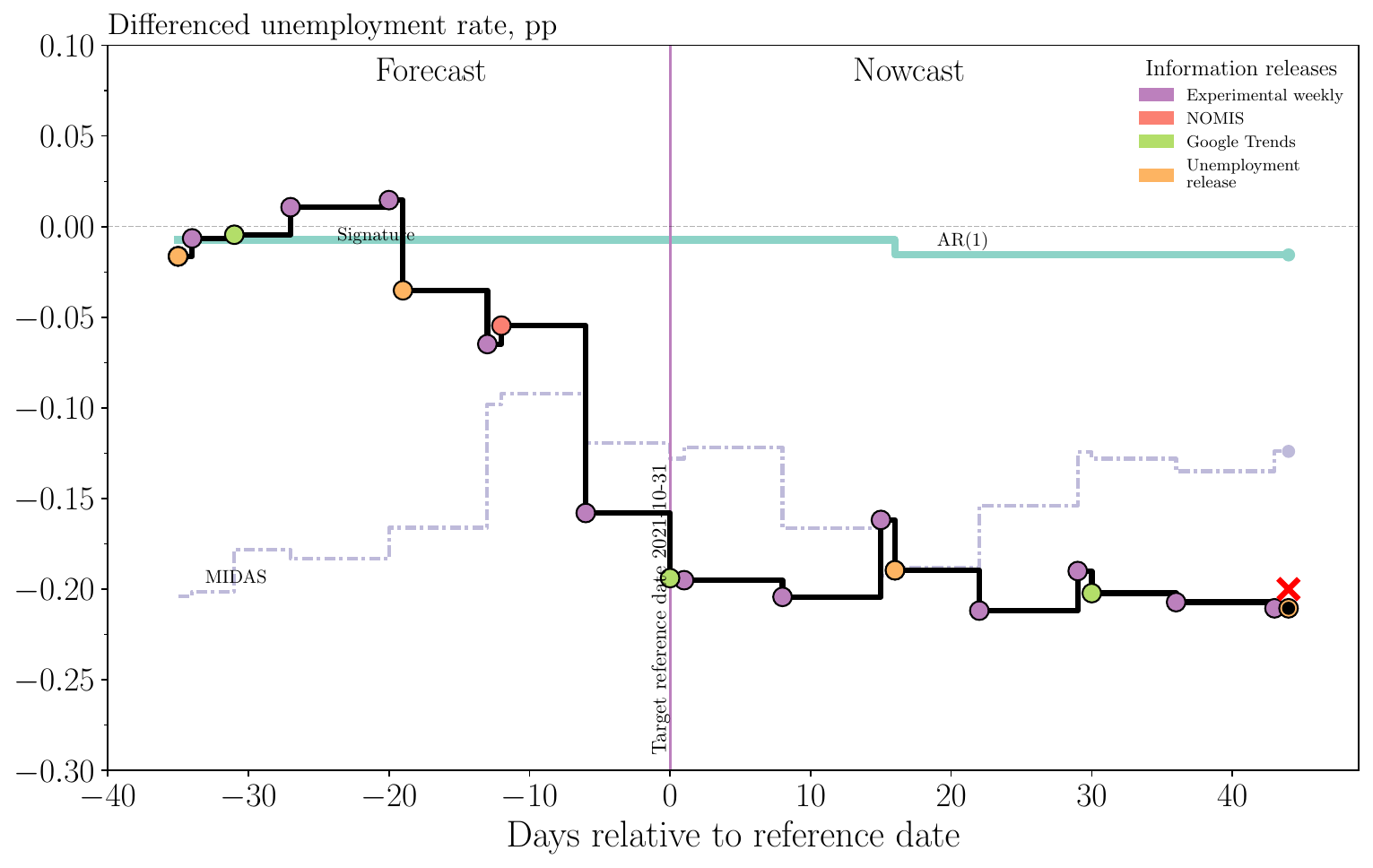}

    \label{fig:lfsexamplecovid}
 \end{figure}

Figure~\ref{fig:lfsexample} shows an example nowcast for the reference date of November 2019. This month saw a significant movement in the monthly unemployment rate, which only the signature model picked up.

Figure~\ref{fig:lfsexamplecovid} shows a nowcast for October 2021. This was a particularly challenging period to nowcast as it was the tail end of the UK's COVID-19 pandemic and the furlough scheme (also known as the Coronavirus Job Retention Scheme) had ended the month before.

In both figures, and in general, the AR(1) updates its prediction only as the previous months' target data become available. The signature method reassesses its prediction with each new information release. Although the MIDAS model also updates with the same frequency and uses the same inputs as the signature model, MIDAS fares only slightly better than the AR(1) in these examples.

In Figure~\ref{fig:lfsrmseavg}, we show just how typical it is that the signature outperforms both the AR(1) and the MIDAS model. Figure~\ref{fig:lfsrmseavg} shows the root mean squared error (RMSE) for all three models averaged over the test period. The signature method consistently outperforms the other two models over relevant publication timeframes, and overall: the average of the RMSEs over the reference periods was $0.163$, $0.175$, and $0.087$ for the AR(1), MIDAS, and signature method respectively. The performance improvement relative to MIDAS is 50 percent. As can be seen in both example figures, the signature method appears particularly strong at catching large swings in the unemployment rate as compared to the other models.

\begin{figure}[h!]
    \centering
    \caption{The average root mean square error of an AR(1) versus a MIDAS model versus the signature method as applied to nowcasting UK unemployment.}     \includegraphics[width=\textwidth]{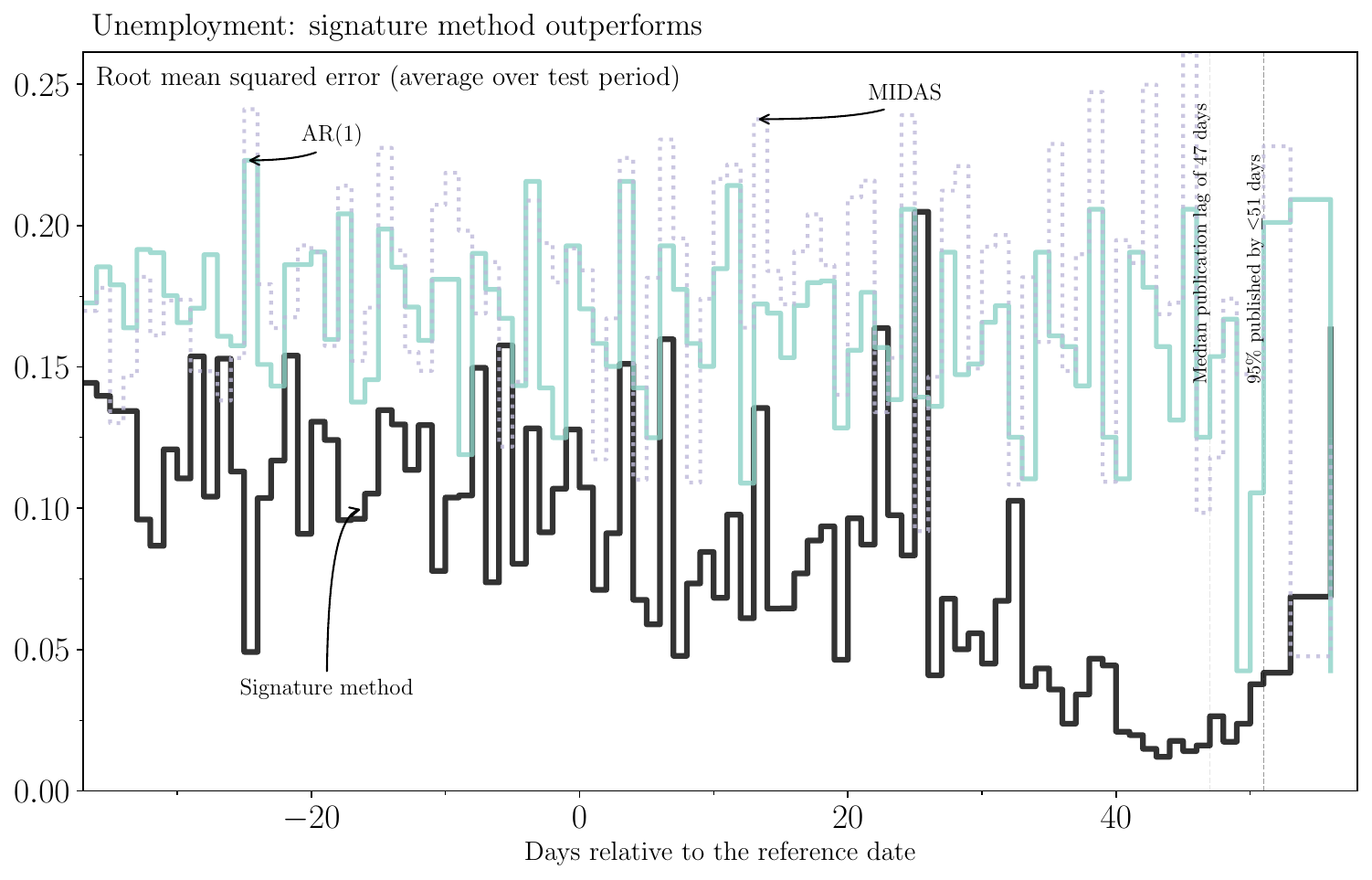}
     \label{fig:lfsrmseavg}
 \end{figure}

There are important conclusions to draw from this exercise. The first is that the signature method is at home with truly difficult time series problems -- here, we tackle irregular series, delayed target variable publication, and mixed frequencies. Second, we see that the signature method outperforms the baselines at almost any point during the tracking period. Third, the signature method still performs well in difficult times -- even during COVID-19, when employment and economic activity were highly disrupted and the swings in unemployment rate were considerable, the signature method provides a stronger nowcast than the other approaches. This matters enormously for policymaking because it is precisely during disrupted times that decisions really count.

%% file: appendices/app_framework.tex
\section{Model Framework}\label{sec:framework}

We set out the proposed model framework in two algorithms.

Algorithm~\ref{alg:sigreg_sub} gives the procedure to fit a single signature regression model on a set of observed data. The full signature regression method pipeline is presented in Algorithm~\ref{alg:sigreg}. Note that while this uses hyper-parameter optimisation as part of the pipeline with grid search, other hyper-parameter search strategies could easily be incorporated.

The full list of configuration settings, and a brief explanation of them, may be found in Appendix~\ref{app:configs}.

\begin{algorithm}[h!]
\SetAlgoLined
\textit{Input}:
\begin{itemize}
    \item data (observation \& target);
    \item target variable name;
    \item data parameters, e.g. publication lag in the target variable;
    \item model parameters associated with signature regression, e.g. the truncation level, lookback window length, whether to use the previous known value of the target etc.; these may be given as lists to perform a hyper-parameter search over.
\end{itemize}

\textit{Process}:
\begin{enumerate}
    \item Define the auxiliary dataframe of available observations.
    \item Convert the path information of the observations into truncated signatures over each lookback window.
    \item Split signatures into training data and test data. The training data can be a proportion of the data available, or it can be the maximum data available (all timepoints for which the target is available).
    \item Fit a (regularised) regression model.
\end{enumerate}
\textit{Return}: fitted model

\caption{Fit signature regression model}\label{alg:sigreg_sub}
\end{algorithm}

\begin{algorithm}[h!]
\SetAlgoLined
\textit{Input}:
\begin{itemize}
    \item data (observation \& target) in a wide table, pivot format;
    \item target variable name;
    \item data parameters, e.g. publication lag in the target variable;
    \item model parameters associated with signature regression, e.g. the truncation level, lookback window length, whether to use the previous known value of the target etc.; these may be given as lists to perform a hyper-parameter search over.
\end{itemize}

Impute missing values in the dataset by the specified method (defaults to forward fill for all values except those at the beginning, which are filled by backward fill).
\newline \newline \quad
\For {each set of hyperparameters/configurations}{
    \For{each time over the hyperparameter optimisation period (validation set)}{
    \begin{enumerate}
        \item Fit a signature regression estimate as outlined in Algorithm \ref{alg:sigreg_sub}.
        \item Evaluate the hidden data set with the regression model.
        \item Store predictions and errors.
    \end{enumerate}
    }
}
\quad \newline \quad
Select the best set of hyperparameters.
\newline \newline \quad
\uIf{Recursive nowcasts}{

\For{each time that a nowcast is required}{
    \begin{enumerate}
    \item Fit a signature regression model as outlined in Algorithm~\ref{alg:sigreg_sub}.
    \item Evaluate the test set (only the time for the current nowcast) with the regression model for a nowcast value.
    \end{enumerate}
}    
}
\Else{
    \begin{enumerate}
    \item Combine the train and validation set and fit a signature regression model as outlined in Algorithm~\ref{alg:sigreg_sub}.
    \item Evaluate the test set with the above trained regression model.
    \end{enumerate}
}
\quad \newline \quad
\textit{Return} nowcast values.    

\caption{Signature regression framework }\label{alg:sigreg}
\end{algorithm}

%% file: appendices/app_simulation.tex
\section{Additional simulation results}\label{app:simulation}

Here, we provide similar charts to Figures~\ref{fig:sim_residuals_composite} and \ref{fig:traj} for the three other simulation cases: irregularly sampled data, data generated by a nonlinear process, and irregularly sampled data generated by a nonlinear process. The former chart type compares the residuals from a signature regression with parameters learned from the training dataset versus a Kalman filter with true parameter values. For the means and variances of residuals, see Table~\ref{tab:simsummary}. The latter plot type gives the (transformed) observed values along with the true, hidden values that need to be inferred, as well as the predicted trajectories from the Kalman filter and the signature method.

\subsection{Irregularly sampled data}

Figure~\ref{fig:diagnostic_res_irregular} shows how the residuals from the signature method compare to those from the Kalman filter for paths in the test set for irregularly sampled data (where 80\% of the data was dropped from the regularly sampled data). 
The line of best fit shown in Figure~\ref{fig:diagnostic_res_irregular} has gradient $0.90$ and intercept $-0.03$. The equivalent irregularly sampled inference of a single path as in Figure~\ref{fig:traj} is shown in Figure~\ref{fig:traj2}: the signature method is able to provide a good approximation to the Kalman filter.

\begin{figure}[h!]    \centering
    \caption{Residuals of the Kalman filter vs the signature method on irregularly sampled paths in the test set. Marginal plots are histograms of residuals. The red line gives the line-of-best-fit.}    \includegraphics[width=0.75\linewidth]{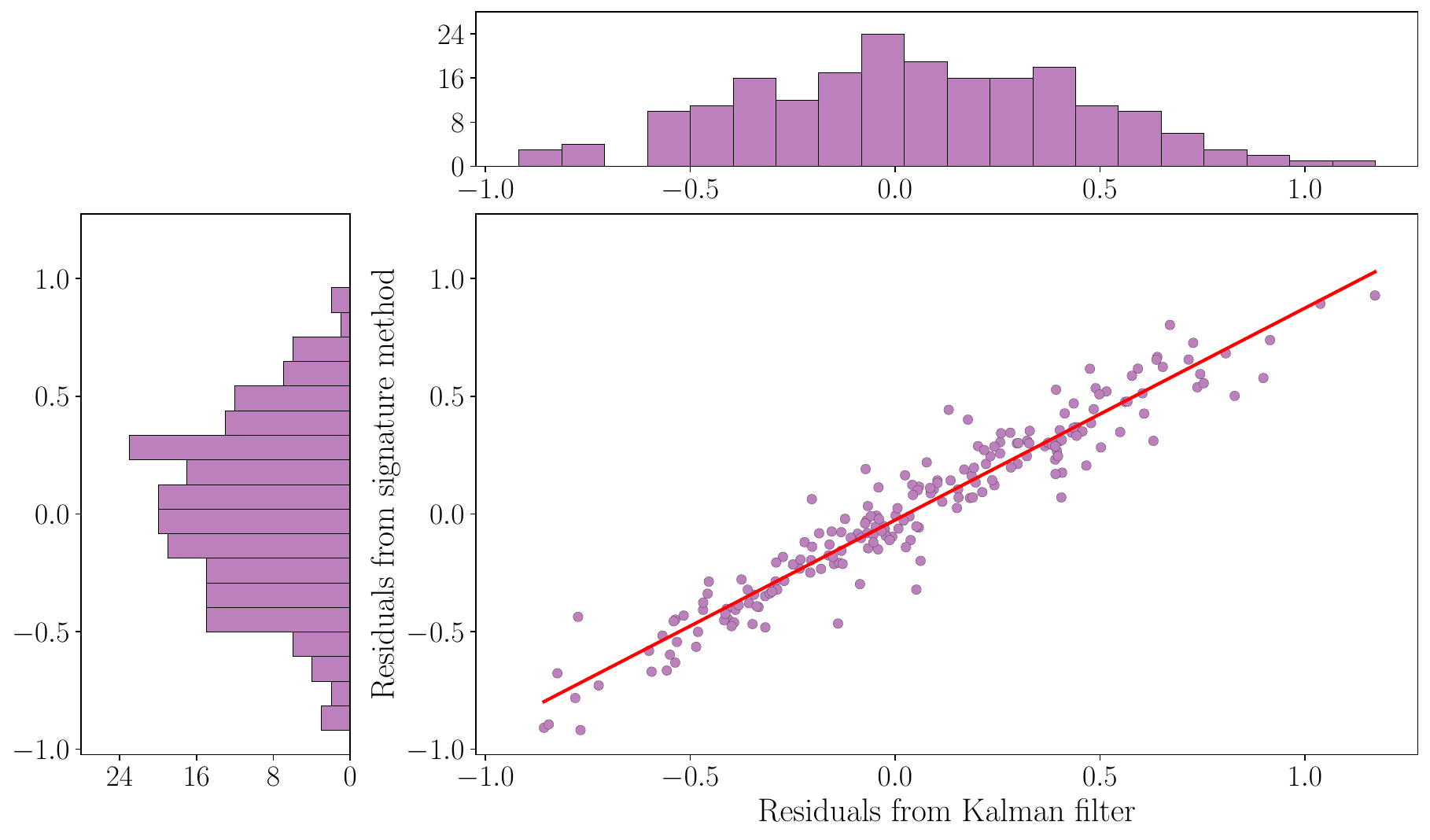}
    \label{fig:diagnostic_res_irregular}
\end{figure}

\begin{figure}[h!]
    \caption{An \textit{irregularly} sampled simulated path with observed values (\(X_t\)) along with the true, hidden values that need to be inferred (\(Y_t\)). }
    \centering
    \includegraphics[width=0.75\textwidth]{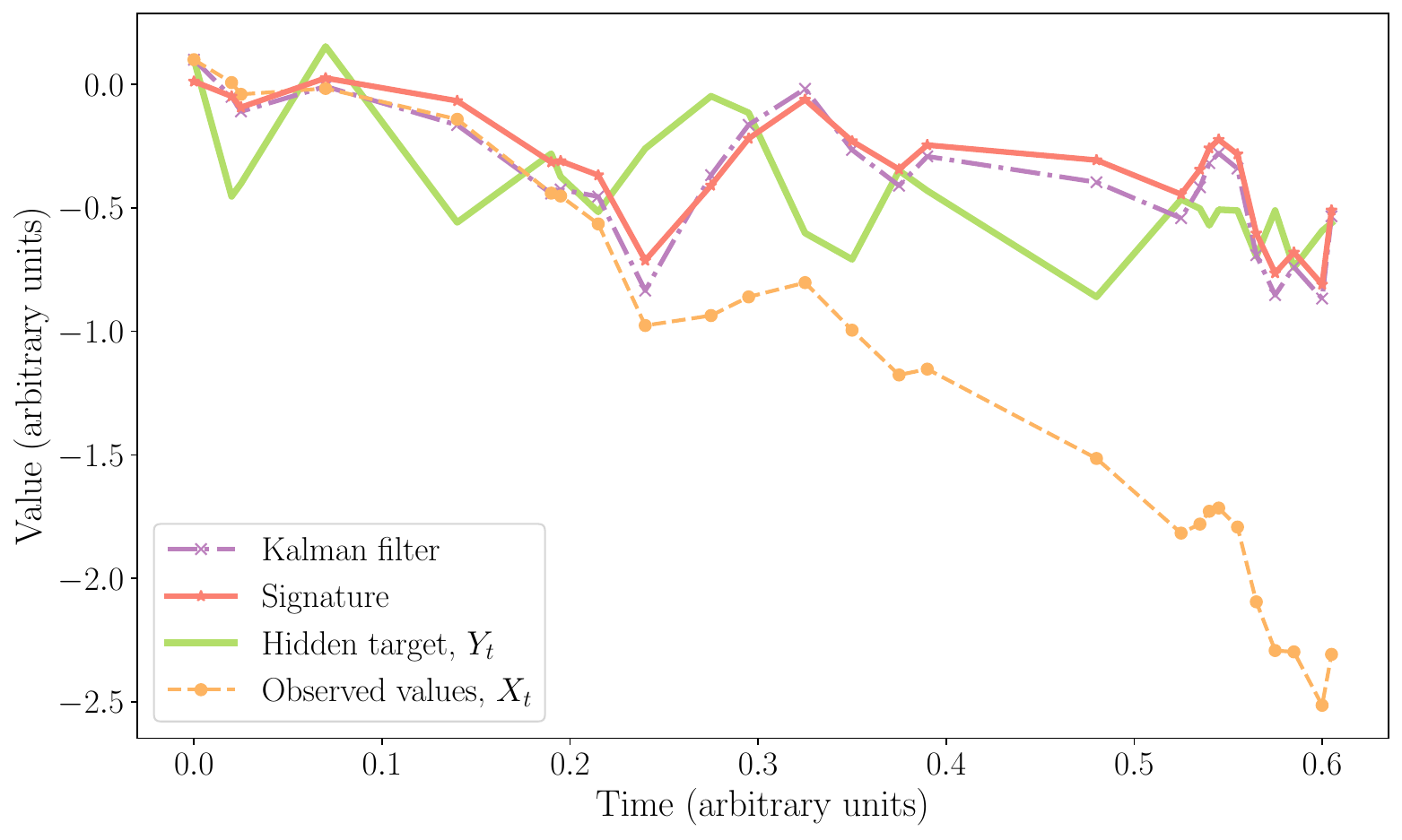}
    \label{fig:traj2}
\end{figure}

The simulation demonstrates that regression on signatures is a competitive method to infer parameters even when data are not regularly sampled. No additional modifications need be made to the signature method to achieve this.

\subsection{Nonlinear data generating process}

We apply a sigmoid transform to obtain a nonlinear process on regularly sampled trajectory data.
This transform allows the (theoretically optimal) Kalman filter to be derived in this nonlinear context, for the sake of comparison.
The residual comparison is in Figure~\ref{fig:diagnostic_res_sigmoid} (gradient $0.82$, intercept $0.01$), while an example of inferred paths on a single trajectory (as in Figure~\ref{fig:traj}) is in Figure~\ref{fig:traj_sigmoid}.

\begin{figure}[h!]    \centering
    \caption{Residuals of the Kalman filter vs the signature method on data generated by a nonlinear process. Marginal plots are histograms of residuals. The red line gives the line-of-best-fit.}    \includegraphics[width=0.75\linewidth]{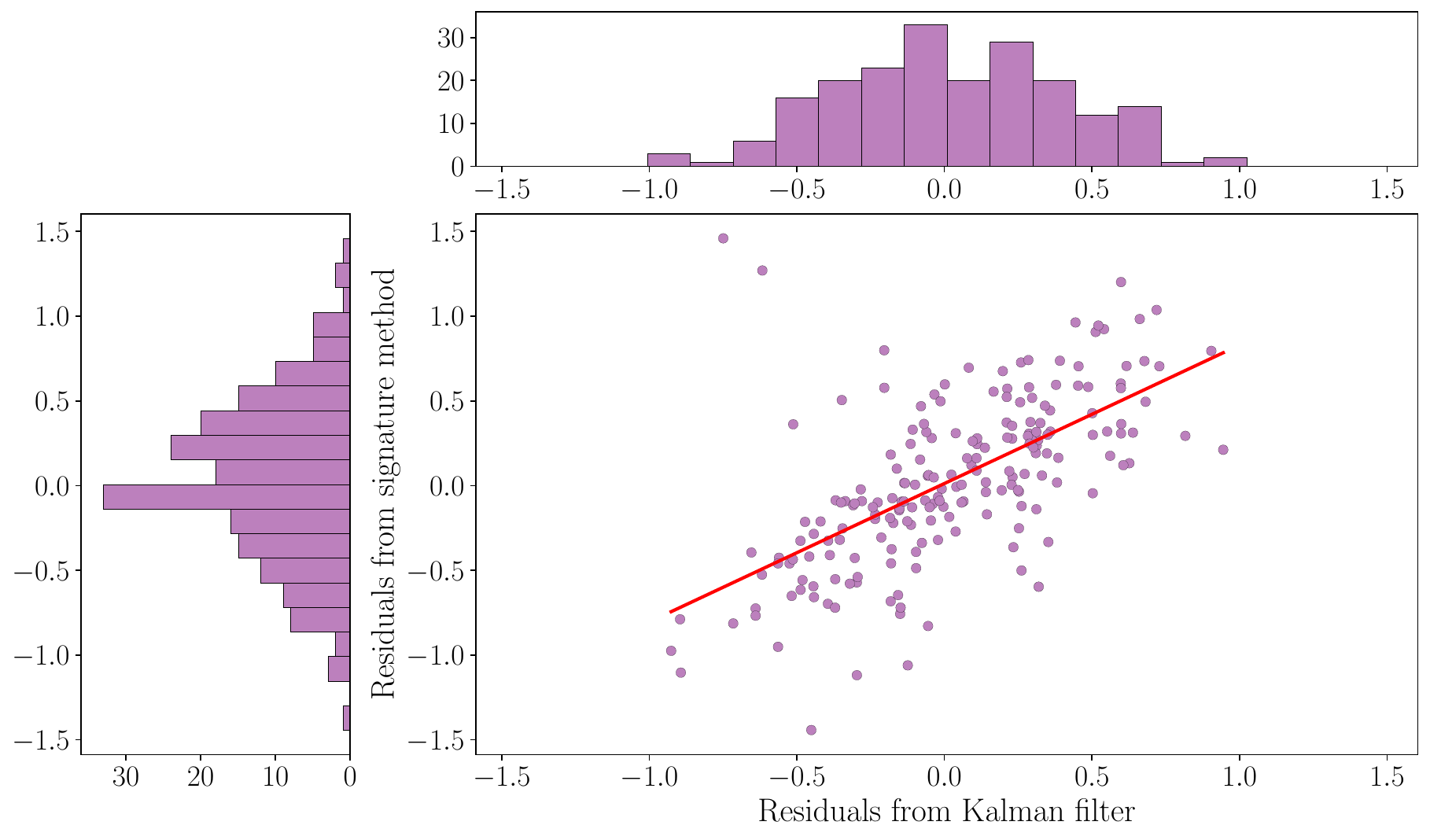}
    \label{fig:diagnostic_res_sigmoid}
\end{figure}

\begin{figure}[h!]
    \caption{Illustration of the same simulated path example as in Figure~\ref{fig:traj} but now the observed data have been passed through a sigmoid transform.}
    \centering
    \includegraphics[width=0.9\textwidth]{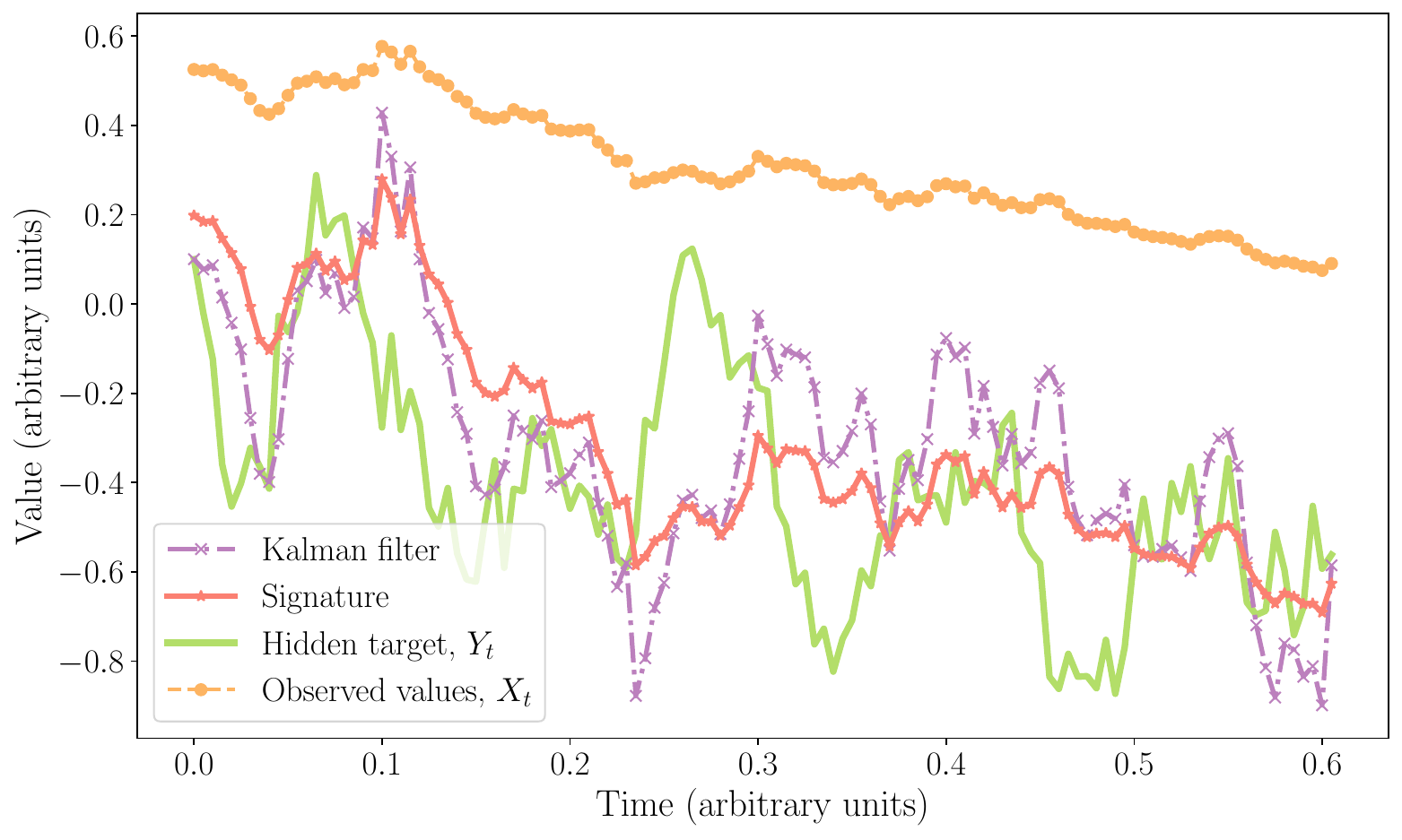}
    \caption*{\small{\textit{Notes.} The plot gives the transformed observed values along with the true, hidden values that need to be inferred, as well as the predicted trajectories from the Kalman filter and the signature method.}}
    \label{fig:traj_sigmoid}
\end{figure}

\subsection{Nonlinear with irregularly sampled data}
We randomly drop 80\% of the data from the nonlinear process to obtain irregularly sampled data.

Figure~\ref{fig:diagnostic_res_sigmoid_down} shows the comparison of the residuals, which has a much weaker relationship in this case (gradient $0.27$, intercept $0.03$). In Figure~\ref{fig:traj_sigmoid2}, we plot the results on the same simulated path example as in Figure~\ref{fig:traj_sigmoid} but with observations that have been irregularly sampled. Despite the weaker correlation between the signature method residuals and the Kalman filter residuals, the predictions are not dissimilar.

\begin{figure}[h!]    \centering
    \caption{Residuals of the Kalman filter vs the signature method on data generated by a nonlinear process and downsampled. Marginal plots are histograms of residuals. The red line gives the line-of-best-fit.}
    \includegraphics[width=0.75\linewidth]{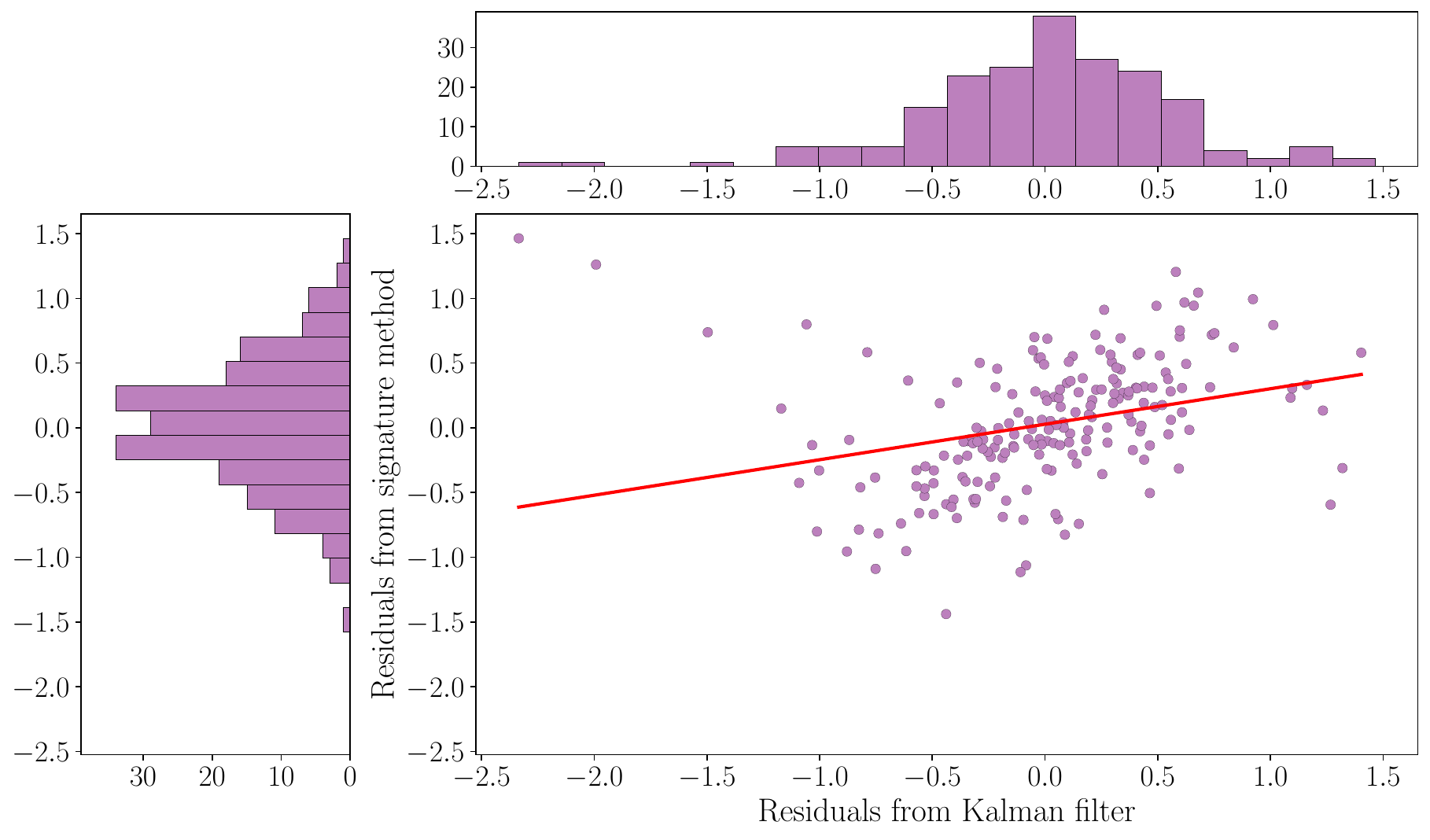}
\label{fig:diagnostic_res_sigmoid_down}
\end{figure}

\begin{figure}[h!]
    \caption{Illustration of the same simulated path example as in Figure~\ref{fig:traj_sigmoid} (that is, after a sigmoid transform) but now irregularly sampled. }
    \centering
    \includegraphics[width=0.9\textwidth]{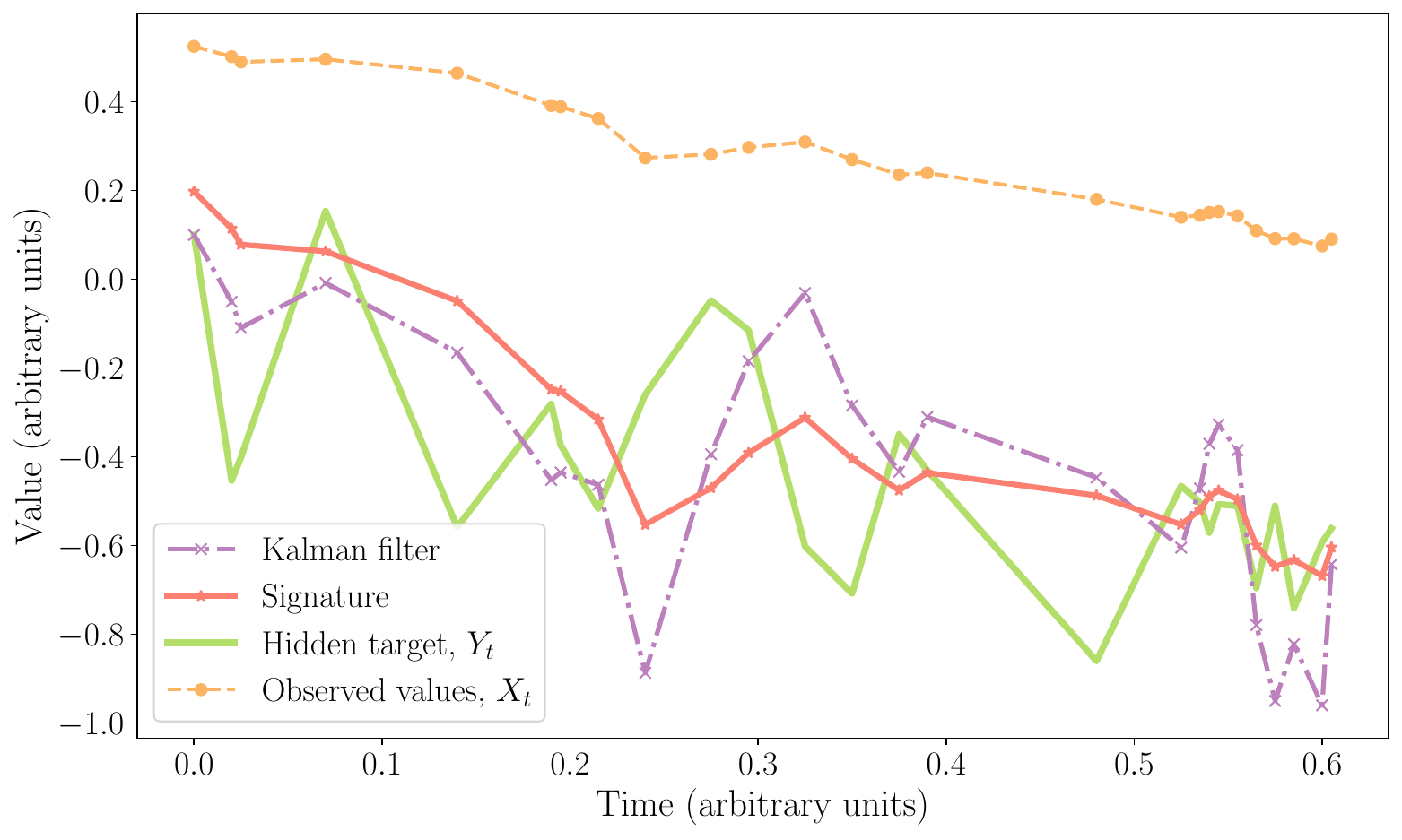}
    \label{fig:traj_sigmoid2}
\end{figure}

%% file: appendices/app_configs.tex
\section{Configurations for nowcasting pipeline}\label{app:configs}

A short description of the main configuration settings for our signature nowcast can be found in Table \ref{tab:config_desc}.

\begin{table}[]
    \centering
    \caption{Description of the main configs used in the nowcasting pipeline.}
    
    \begin{tabular}{c c m{18em}}
    
        Parameter & Data Type & Description \\
        \hline
        \hline
        window\_type & str & The type of window used for generating the signatures -- possible values are `days', `ind' and None. Arguments `days' and `ind' behave as moving windows. None acts as an expanding window. \\
        \hline
        max\_length & int & When `window\_type' is set to `days' or `ind', `max\_length' is the number of days or rows to index by.\\
        \hline
        fill\_method & str & How missing data should be filled - typical values `ffill', `bfill' or `rectilinear'. \\
        \hline
        level & int & Truncation level for the variables.\\
        \hline
        t\_level & int & Truncation level for the time parameter.  \\
        \hline
        basepoint & bool & Whether to add a basepoint of zeros before computing signatures to remove lateral invariance.\\
        \hline
        use\_multiplier & bool & Whether to multiply only the signatures corresponding to time only by multiplier based on target series. \\
        \hline
        keep\_sigs & str & The type of signature terms to keep - values `all', `innermost', `all\_linear' and None.\\
        \hline
        regularize & str & Method to regularize, accepted values are `elastic\_net', `l1', `l2', and None. Default: `elastic\_net'.\\
        \hline
        alpha & float & Determines the strength of regularisation, see \href{https://scikit-learn.org/stable/modules/generated/sklearn.linear_model.ElasticNet.html}{\texttt{sklearn}}.\\
        \hline
        l1\_ratio & float & Weight for the L1 penalty in the ElasticNet of \href{https://scikit-learn.org/stable/modules/generated/sklearn.linear_model.ElasticNet.html}{\texttt{sklearn}}.\\
        \hline
        fit\_intercept & bool & Whether to fit an intercept in regression.\\
        \hline
        target\_lag & int & The target variable will have a lag due to the release timetable, this variable creates a shift so models never utilise future information.\\
        \hline
        use\_prev\_value & bool & Whether to use the latest value of the target series for the regression. \\
        \hline
        standardize & bool & Standarize signature terms before regression. Default: True. \\
        \hline
        training\_proportion & float & This is an optional parameter to specify a fixed proportion for the train set.\\
        \hline
        reduce\_dim & bool & Whether to reduce the dimension of the observed variables with principal component analysis.\\
        \hline
        k & int & The number of principal components to use.\\
        
    \end{tabular}
    \label{tab:config_desc}
\end{table}

In addition, the pipeline includes the possibilities to use custom factor structures to find principal components and much more.

%% file: appendices/app_usgdp_factors.tex
\section{Factor structure for the US GDP DFM}\label{sec:usgdp_factors}

As in \citet{bok2018macroeconomic}, we assume that there are 4 factors labelled as soft, real, labour, and global which affect the time series behind the US GDP nowcasting application. In this appendix, we list the groupings. The global factor includes all variables (i.e. it contains factors in soft, real, and labour as well). The final list below, ``other'', lists all the variables which do not fall under soft, real, and labour, and so are affected by the global factor alone.

\vspace{1em}

\noindent \textbf{Soft} = \{Empire state Mfg. Survey: General business conditions, Philly Fed Mfg. business outlook: Current activity, University of Michigan: Consumer Sentiment\}

\vspace{1em}

\noindent \textbf{Real} = \{Real domestic gross product, Manufacturers new orders: Durable goods, Retail sales and food services, New single family houses sold, Housing starts, Industrial production index, Merchant wholesalers: Inventories: Total, Value of construction put in place, Building permits, Capacity utilization, Inventories: Total business, Real personal consumption expenditures, Manufacturing shipments: Durable goods, Mfrs. unfilled orders: All manufacturing industries, Manufacturing inventories: Durable goods, Real gross domestic income, Real disposable personal income, Exports: Goods and services, Imports: Goods and services\}

\vspace{1em}

\noindent \textbf{Labour} = \{All employees: Total nonfarm, Civilian Unemployment rate, ADP nonfarm private payroll employment, Nonfarm business sector: Unit labor cost, JOLTS: Job openings: Total\}

\vspace{1em}

\noindent \textbf{Other} = \{CPI-U: All items, PPI: Final demand, Import price index, PCE less food and energy: Chain price index, CPI-U: All items less food and energy, PCE: Chain price index, Export price index\}

%% file: appendices/app_exp_details.tex
\section{US GDP Specification Details}\label{sec:hyperparams_us}

For the US GDP growth example, there are two models, depending on whether the data used are the principal components of each factor group, or whether they are the final hidden factors from the dynamic factor model (with factor structure as described in Appendix~\ref{sec:usgdp_factors}).

The target variable has a 30 day publication lag after the end of the quarter. This corresponds to a lag of about 124 days from the start of the reference quarter. All signature terms are standardised as discussed in Section \ref{sec:regularisation}.

For the data obtained using PCA, the final fitted model used all the linear signature terms truncated at level 4 over a lookback window of 730 days. This is a ridge model with a regularisation parameter strength of 2.0. We use the previous value of the target in the regression and allow an intercept. We use the target series as a multiplier as seen in equation \eqref{eqn:regressionsig} but we do not append a basepoint to the start of the series (this is where a value, often 0, is appended to remove the lateral invariance) when computing the signatures. Missing data is filled by forward fill (carry the last value forward until another observation is made). A summary of these hyperparameters can be found in Table \ref{tab:params_gdp_pca}.

\begin{table}[hp]
    \centering
    \caption{Table of hyperparameters used to obtain the US GDP results with PCA ``factors''. Items in bold have been selected by tuning; otherwise we have chosen their values.}
    \begin{tabular}{c c}
         Parameter & Value \\
         \hline
         \hline
        level & \textbf{4} \\
        t\_level & \textbf{3} \\
        fill\_method & \textbf{ffill} \\
        max\_length & \textbf{730} \\
        keep\_sigs & all \\
        regularize & elasticnet \\
        alpha & \textbf{2.0} \\
        l1\_ratio & \textbf{0} \\
        standardize & True \\
        fit\_intercept & True \\
        use\_multiplier &  True \\
        basepoint &  \textbf{False} \\
        target\_lag & 124 \\
        use\_prev\_value & True
        
    \end{tabular}
    \label{tab:params_gdp_pca}
\end{table}

For the data obtained using DFM means as observations, the final fitted model used all the linear signature terms truncated at level 4 over a lookback window of 730 days. This is an elastic net model with a regularisation parameter strength of 0.5. We use the previous value of the target in the regression and allow an intercept. We use the target series as a multiplier but we do not use a basepoint when computing the signatures. Missing data is filled by forward fill interpolation. A summary of these hyperparameters can be found in Table \ref{tab:params_gdp_dfm}.

\begin{table}[h]
    \centering
    \caption{Table of hyperparameters used to obtain the US GDP results with filtered DFM factors. Items in bold have been selected by tuning; otherwise we have chosen their values.}
    \begin{tabular}{c c}
         Parameter & Value \\
         \hline
         \hline
        level & \textbf{4} \\
        t\_level & \textbf{2} \\
        fill\_method & \textbf{ffill} \\
        max\_length & \textbf{730} \\
        keep\_sigs & all \\
        regularize & elasticnet \\
        alpha & \textbf{0.5} \\
        l1\_ratio & \textbf{0.5} \\
        standardize & True \\
        fit\_intercept & True \\
        use\_multiplier &  True \\
        basepoint &  \textbf{False} \\
        target\_lag & 124 \\
        use\_prev\_value & True
        
    \end{tabular}
    \label{tab:params_gdp_dfm}
\end{table}

\section{UK unemployment hyperparameters}\label{app:lfshyper}

The list of hyperparameters selected for nowcasting the UK unemployment can be found in Table \ref{tab:params_lfs}.

\begin{table}[h]
    \caption{Table of hyperparameters used for nowcasting UK unemployment -- see Section~\ref{sec:lfsapplication}. Items in bold have been selected by tuning; otherwise we have chosen their values.}
    \centering
    \begin{tabular}{c c}
         Parameter & Value \\
         \hline
         \hline
        level & \textbf{4} \\
        t\_level & \textbf{3} \\
        fill\_method & rectilinear \\
        max\_length & 40 days \\
        keep\_sigs & all\_linear \\
        regularize & elasticnet \\
        alpha & \textbf{0} \\
        l1\_ratio & \textbf{0} \\
        standardize & \textbf{False} \\
        fit\_intercept & True \\
        use\_multiplier &  False \\
        basepoint & \textbf{True} \\
        target\_lag & 8 \\
        use\_prev\_value & True
        
    \end{tabular}

    \label{tab:params_lfs}
\end{table}

%% file: appendices/app_lfs.tex
\section{Robustness of unemployment results to different test set sizes}
\label{sec:lfsrobustness}

A training (and validation) window of 80\% of the data is a convention in machine learning but using a single window may not give a good indication of predictive accuracy, as described in \citet{rossi_outofsample_2012}. Therefore, in this section, we perform a sweep of training and validation set window sizes, with the proportion ranging from $\underline{p}=0.15$ to $\overline{p}=0.85$ of the full dataset. We follow \citet{rossi_outofsample_2012}'s formulation, and ask whether, across these sweeps, the advantage of the signature nowcasting method in RMSE appears i) at any set size and ii) on average.

Let $T$ be the total number of nowcasting periods. We take a sweep of $N$ estimations between
\[
\underline{R} = \lfloor{\underline{p} \cdot T\rfloor} \text{ and } \overline{R}=\lfloor{\overline{p}\cdot T\rfloor}.
\]

This gives $\underline{R} = 86$ and $\overline{R} = 154$, and $N = 69$ different estimations. Let $M=T-R$, i.e. let $M$ be the size of the test set. Also, let us denote distinct trained models as $f_{1}(t, R)$, $f_2(t, R)$, etc., for a given $R$ and test-set nowcast $R \leq t \leq T$.

For each model pair, compute the per-observation loss differential as
\[
\Delta L_{t}(\hat f_1(t, R), \hat f_2(t,R)) = (y(t) - \hat y_{1}(t))^2 - (y(t) - \hat y_2(t))^2.
\]

Then for each $R \in \{\underline{R}, \dots, \overline{R}\}$, the standardised loss differential is
\[
\Delta L_T(R) = \frac{1}{\hat\sigma_R} \, M^{-1/2} \sum_{t=R}^{T} \Delta L_{t}(\hat f_1(t,R), \hat f_2(t,R)),
\]
where
\[
\hat\sigma_R^2 = \sum_{i=-q(M)+1}^{q(M)-1} \left(1 - \left|\frac{i}{q(M)}\right|\right) M^{-1} \sum_{t=R}^{T} \Delta L^d_{t}(\hat f_1(t,R), \hat f_2(t,R)) \, \Delta L^d_{t-i}(\hat f_1(t-i,R), \hat f_2(t-i,R))
\]
and the demeaned differential loss is
\[
\Delta L^d_{t}(\hat f_1(t, R), \hat f_2(t, R)) = \Delta L_{t+h}(\hat f_1(t, R), \hat f_2(t, R)) - M^{-1} \sum_{t=R}^{T} \Delta L_{t}(\hat f_1(t, R), \hat f_2(t, R)).
\]
We use the rule-of-thumb $q(M)=\lfloor{4 \left(\frac{M}{100}\right)^{2/9}\rfloor}$ for the bandwidth.

From the above, we compute two quantities of interest for all model pairs. The first measures the maximum advantage across all test set sizes,
\[
\mathcal{R}_T = \sup_{R \in \{\underline{R}, \ldots, \overline{R}\}} \left| \Delta L_T(R) \right|
\]
and the second measures the average advantage across all test set sizes,
\[
\mathcal{A}_T = \frac{1}{\overline{R} - \underline{R} + 1} \sum_{R = \underline{R}}^{\overline{R}} \left| \Delta L_T(R) \right|.
\]

We use the test statistics provided by \citet{rossi_outofsample_2012} to assess the size of $\mathcal{R}_T$ and $\mathcal{A}_T$ in our case, opting for the one-sided tests given that we are checking whether the signature model has better performance. Critical values in this case are 2.13 (10\%), 2.46 (5\%), and 3.11 (1\%) for $\mathcal{R}_T$ and 1.13 (10\%), 1.45 (5\%), and 2.07 (1\%) for $\mathcal{A}_T$.

We compare the signature model pairwise against MIDAS and AR(1), and also MIDAS against AR(1). Table~\ref{tab:lfs-robustness} reports the resulting $\mathcal{R}_T$ and $\mathcal{A}_T$ statistics. From the pairwise comparisons in Table~\ref{tab:lfs-robustness}, we see that signatures have distinct and superior performance at the 1\% significance level. Figure~\ref{fig:robustness-lfs} shows the absolute differential error between the same pairs of models and Figure~\ref{fig:robustness-lfs-signed} shows its signed equivalent.

In every configuration, the signature method has an advantage over the other models. The only configurations where this is not emphatically the case are when there are very few training data points.

\begin{table}[]
    \centering
    \begin{tabular}{lll}
    \toprule
     & $\mathcal{R}_T$ & $\mathcal{A}_T$ \\
    Pair &  &  \\
    \midrule
    Signature vs MIDAS & 4.15*** & 3.00*** \\
    Signature vs AR(1) & 4.33*** & 3.07*** \\
    MIDAS vs AR(1) & 4.00*** & 2.76*** \\
    \bottomrule
    \end{tabular}
    \caption{Pairwise advantage statistics $\mathcal{R}_T$ and $\mathcal{A}_T$.}
    \label{tab:lfs-robustness}
\end{table}

\begin{figure}[h]
    \centering
    \includegraphics[width=0.9\textwidth]{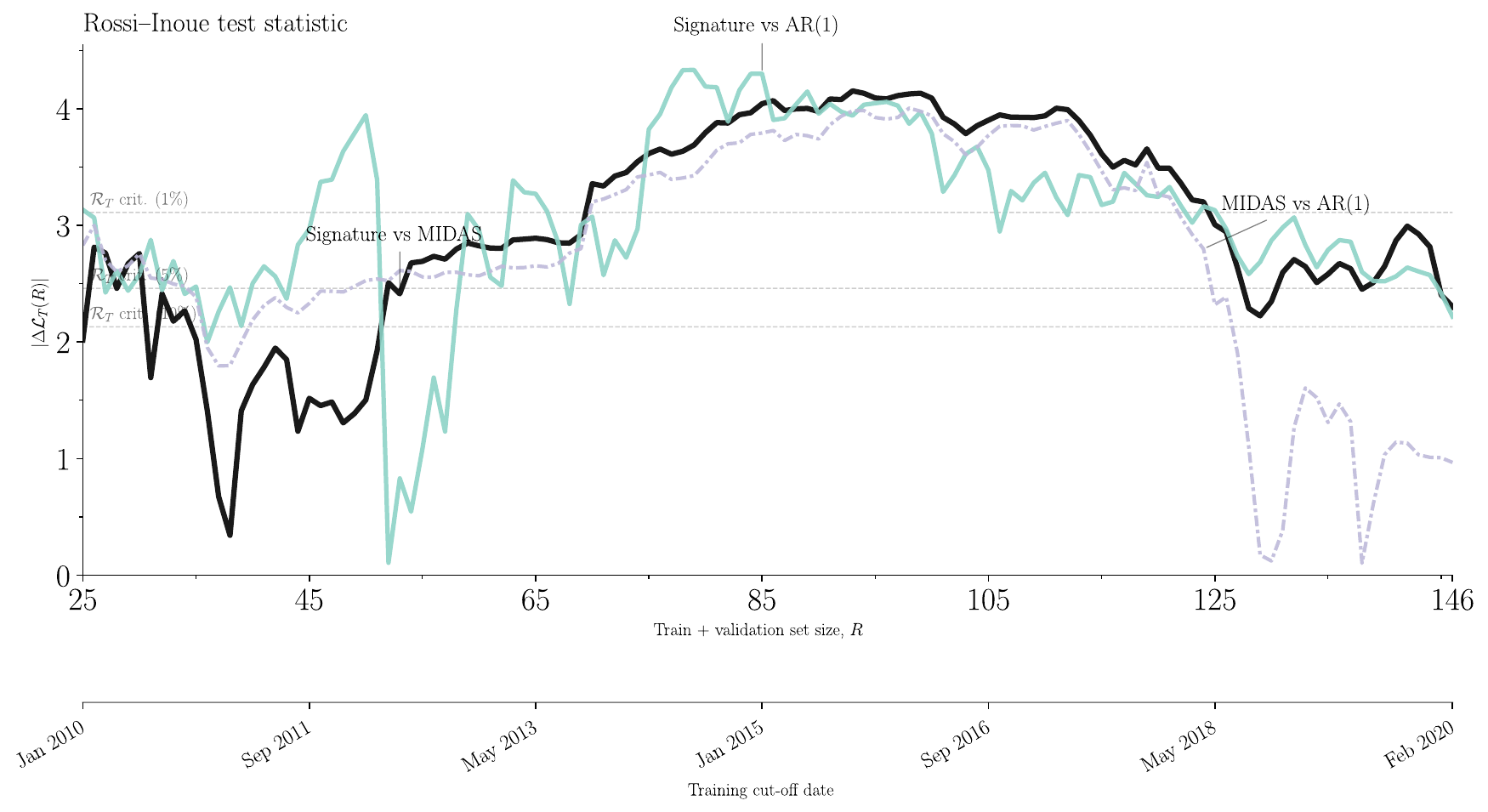}
    \caption{Absolute differential error between pairs of models.}
    \label{fig:robustness-lfs}
\end{figure}

\begin{figure}[h]
    \centering
    \includegraphics[width=0.9\textwidth]{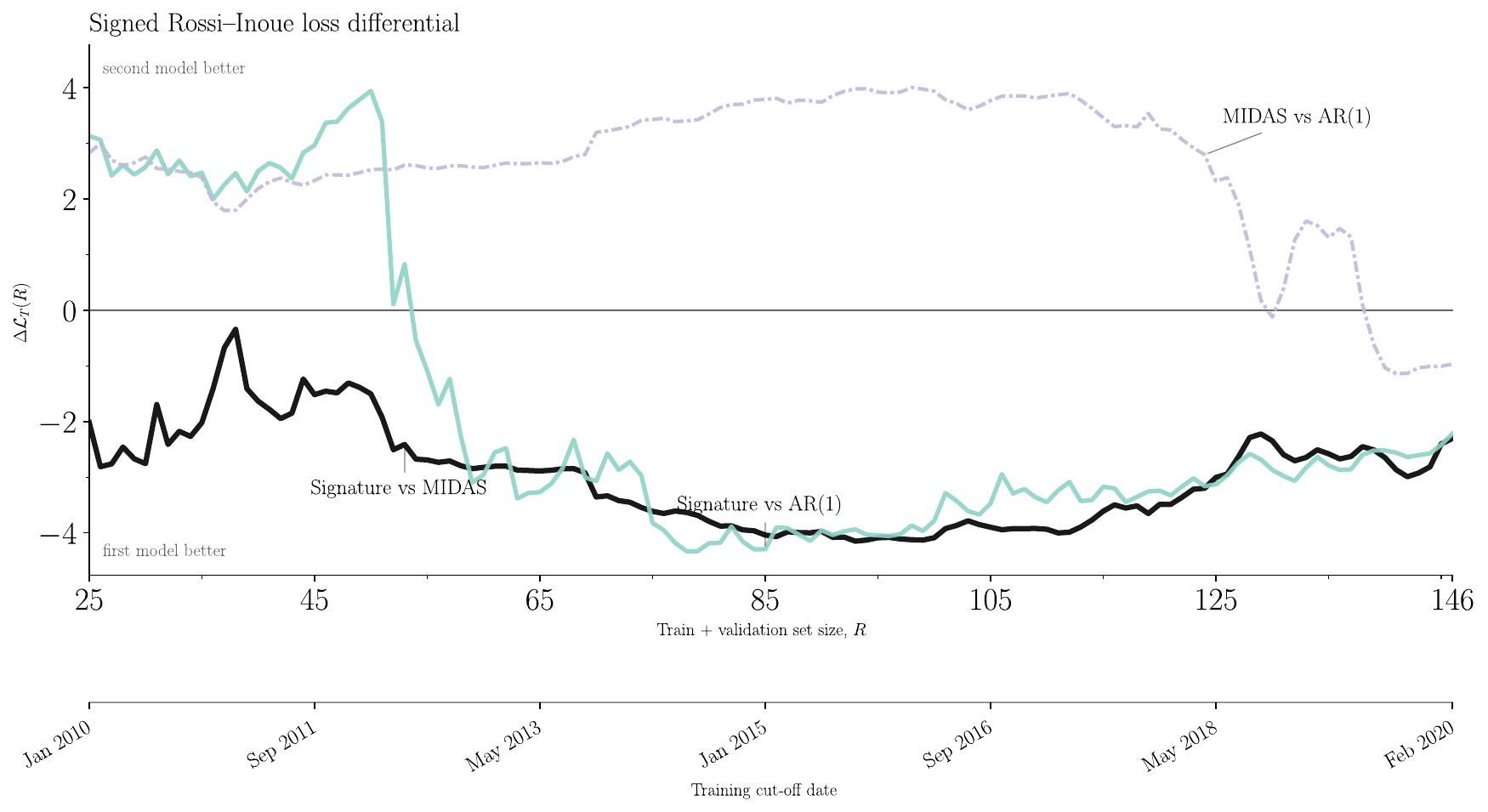}
    \caption{Signed differential error between pairs of models.}
    \label{fig:robustness-lfs-signed}
\end{figure}

%% file: appendices/app_consistency.tex
\section{Consistency of regression on signatures}\label{app:consistency}

The consistency of the signature regression approach is somewhat delicate, as it naturally involves a high dimensional and possibly ill-posed estimation problem. In this appendix, we will outline what can easily be said about the statistical properties of our method, and give references to where these questions are pursued further in the literature.

Our approach will involve decomposing the errors of predictions into different sources, and showing that each of these can be made small, under appropriate assumptions. We will not give a formal statement of consistency (as this would involve heavy assumptions), but will indicate what would be involved in establishing such a result.

We will require some additional notation in order to perform these decompositions. We recall that $\mathcal{F}^X_t$ is the $\sigma$-algebra describing the information available at time $t$ from observations of $X$ up to time $t$. We write $\mathcal{F}^X_{[t-\tau, t]} = \sigma(\{X_s\})_{s\in [t-\tau, t]}$ for the $\sigma$-algebra of information available from observations of $X$ between times $t-\tau$ and $t$, and $\mathcal{S}^k_{[t-\tau, t]}(X)$ for the $k$th level signature of the augmented process $(t,X_t)$ between times $[t-\tau, t]$. We write $L^*\big(Y_{t-},\mathcal{S}^k_{[t-\tau, t]}(X)\big)$ for the optimal (i.e.~minimising the squared prediction loss) estimator when approximating $Y$ as a linear function of $Y_{t-}$ and $\mathcal{S}^k_{[t-\tau, t]}(X)$. Finally, we write $\hat Y_t = \hat L\big(Y_{t-},\mathcal{S}^k_{[t-\tau, t]}(X)\big)$ for the regression estimate we obtain using the method described in Section~\ref{sec:sigregression}.

We now observe that we can decompose the mean squared error of our estimate as follows (noting that by the orthogonality of conditional-expectation, the first two decompositions do not introduce cross-terms, and the factor $2$ arises from Young's inequality, $(A+B)^2 \leq 2(A^2+B^2)$):
\begin{align*}
    \mathbb{E}[(Y-\hat Y)^2] &\leq\mathbb{E}\Big[(Y-\mathbb{E}[Y|\mathcal{F}_t^X, Y_{t-}])^2\Big] + \mathbb{E}\Big[(\mathbb{E}[Y|\mathcal{F}_t^X,  Y_{t-}] - \mathbb{E}[Y|\mathcal{F}^X_{[t-\tau, t]},  Y_{t-}])^2\Big]\\
    &\quad + 2 \mathbb{E}\Big[(\mathbb{E}[Y|\mathcal{F}^X_{[t-\tau, t]},  Y_{t-}] -L^*(Y_{t-},\mathcal{S}^k_{[t-\tau, t]}(X)) )^2\Big]\\
    &\quad +  2\mathbb{E}\Big[\big(L^*(Y_{t-},\mathcal{S}^k_{[t-\tau, t]}(X))  -\hat Y_t\big)^2\Big].
\end{align*} 
We will address each of these terms in turn.

\subsection{Unavoidable error} 
The quantity
\[\mathbb{E}\Big[(Y-\mathbb{E}[Y|\mathcal{F}_t^X, Y_{t-}])^2\Big]\]
is the unavoidable error of the estimation problem. This cannot be reduced using any mathematical or statistical technique, but is intrinsic to the information available. For this reason we focus on the subsequent errors.

\subsection{Horizon-induced error} 
The quantity
\[ \mathbb{E}\Big[(\mathbb{E}[Y|\mathcal{F}_t^X, Y_{t-}] - \mathbb{E}[Y|\mathcal{F}_{[t-\tau, t]}, Y_{t-}])^2\Big]\]
represents the error induced by restricting our attention to a finite horizon/lookback window. Suppose we know that $Y$ is a Markov process, and assume that $Y_{t-}$ (which is an observation of $Y$ at the most recent point prior to time $t$) comes from an observation time in the interval $[t-\tau, t]$. Then, by standard properties of Markov processes, we know that $Y$ is conditionally independent of $\mathcal{F}^X_{t-\tau}$ given $Y_{t-}$. Recalling that $\mathcal{F}^X_t = \mathcal{F}_{t-\tau}^X \vee \mathcal{F}^X_{[t-\tau,t]}$, we conclude that this error is zero when $Y$ is a Markov process and $\tau$ is sufficiently large relative to the delay in (true) observations of $Y$.

If we do not wish to assume that $Y$ is Markov, or wish to exclude previous observations of $Y$, then this error can be controlled by studying the horizon dependence and ergodic properties of the filter estimate $\mathbb{E}[Y|\mathcal{F}_{[t-\tau, t]}]$. This is a well-studied problem, see for example \cite{vanHandel2014}, \cite{chigansky}, \cite{fausti} and references therein. However, the assumptions under which one can prove bounds on the horizon dependence are generally moderately restrictive.

\subsection{Signature-approximation error} 
The quantity
\[\mathbb{E}\Big[(\mathbb{E}[Y|\mathcal{F}^X_{[t-\tau, t]},  Y_{t-}] -L^*(Y_{t-},\mathcal{S}^k_{[t-\tau, t]}(X)) )^2\Big]\]
represents the error due to approximating the conditional expectation as a linear function of the $k$th level signature. Proving bounds on this approximation error is an active area of mathematical research. The challenge here is to ensure that the map $X \to \mathbb{E}[Y|\mathcal{F}^X_{[t-\tau, t]}, Y_{t-}]$ is continuous in an appropriate sense. If we assume that our observation process $X$ is in discrete time, and that the joint process $(X,Y)$ has some minor level of regularity (for example, it would be sufficient if the vectors $\{X_{t_i}, Y_{t_i}\}_{i=1}^N$ have a joint density), then we can ensure this continuity. Consequently, the classical universal approximation theorem (see \cite{levin2013learning}) applies, which guarantees we can approximate the conditional expectation arbitrarily well with high probability. More recent results \citep{cuchiero2025global, ceylanuniversality} show precisely that the variance of this approximation can be made arbitrarily small, by using a sufficient number of signature terms.

The assumption of linearity in $Y_{t-}$ is not a concern if the true relationship is linear, otherwise this potentially introduces another error in the model. To alleviate this error, one can simply use a richer regression family (for example, polynomial regression, splines, or other nonparametric methods). From this perspective, the assumption on linearity in $Y_{t-}$ is not intrinsic to the signature regression model.

\subsection{Regression error}\label{sec:reg_error}
The quantity
\[ \mathbb{E}\Big[(L^*(Y_{t-},\mathcal{S}^k_{[t-\tau, t]}(X)) -\hat Y_t)^2\Big]\]
represents the error due to the use of regression estimates based on a finite estimation sample.

The error due to regression can now be treated as in a standard regression problem. Observe that we have reduced the problem to a standard finite-dimensional regression problem, for which there is a well-established literature. We outline some primary concerns below.

\paragraph{Regression error with independent observations}

If we supposed that we had access to a sequence of training paths, which were independently sampled, then we would be in a standard regression setting with i.i.d.~observations. The consistency follows from the standard properties of linear regression where, as the sample size increases, the estimated regression coefficients converge to the coefficients in the best linear approximation. These classic results can be found in \citep{drygas1971consistency, drygas1976weak, amemiya1985advanced}, and the result on consistency of regression estimators is fundamental.

The presence of collinearity in the signature terms does raise some concerns here. Certain terms in the signature are guaranteed to satisfy quadratic shuffle identities (for example, the relationships $S^2(tx) + S^2(xt) = S^1(t)S^1(x)$ and $S^2(tt)+ S^2(tx) + S^2(xt) + S^2(xx) = 1/2(S^1(t)+ S^1(x))^2$ are illustrated in Figure \ref{fig:sig_illustration}), and so the regression technique needs to be modified to deal with this issue. Whether this is a practical concern depends on the interpretation of these coefficients -- as the signature is guaranteed to satisfy these relationships, the `true' coefficients are somewhat meaningless, as the model is not fully identifiable (but the predictions from the model are unique, both in- and out-of sample). As the primary concern of these techniques is the model predictions, rather than the interpretation of specific coefficients of individual signature terms, this limitation has few consequences.

In the algorithm presented above, the regression problem is perturbed away from OLS methods by including a regularization term, yielding the regularised least-squares problem, which is strictly convex and hence has a unique optimiser. This may introduce a small amount of bias in the estimation procedure, but the standard results on the ridge regression will apply (see, for example, \cite{hsu2012random}).

We also consider an alternative solution, where we include only the `linear signature' terms, i.e. where time is the only coefficient that is permitted to appear more than once in the signature. This lessens the problem of collinearity, provided we use data with a range of time horizons. The drawback of this restriction is that it forces a linear relationship to hold between the hidden target variable and the observation paths, which may be inappropriate in some contexts. We see that this assumption naturally arises in Section \ref{sec:signKalman}, which is based on classical (linear) Kalman filtering.

\paragraph{Regression error with dependent observations}

The fact that our method reduces to classical linear regression enables easy extension from the traditional setting where we have access to independent observation paths, and allows more realistic time-series assumptions to be made instead.

A classic extension is to assume that our observations are from an ergodic process, for which the use of a single long trajectory results in estimates that converge to their underlying averages. In particular, \cite{lucchese2025learning} show that, under the assumption that our observation process is strongly mixing, the signature of our observations will also be strongly mixing (curiously, they also demonstrate that simply assuming ergodicity of the observation process is not sufficient to guarantee ergodicity of its signature). They then show the consistency of the sample average signature using non-overlapping intervals (along with central limit theorems for these quantities).

In a regression context, this implies that, under the assumption that our observations $X$ and hidden target $Y$ are (jointly) strongly mixing, the regularised regression estimator will also be consistent\footnote{The fact that \cite{lucchese2025learning} only show consistency of the average signature is not a concern here, as Chen's identity shows that products of signature terms are themselves signature terms. This allows us to expand the regression coefficient in terms of the joint signature of $(X,Y)$ over each non-overlapping interval, and hence apply their result to see consistency.}.

When we use overlapping intervals, the results we obtain do not fall neatly into the setting of \cite{lucchese2025learning}. However, we can consider the non-overlapping interval case as simply being the combination of a collection of correlated learning problems with non-overlapping intervals. As all these problems are consistent, we see that this does not cause significant difficulty for the theory.

\paragraph{Infill asymptotics}

A further asymptotic regime that could be considered, but is not the focus here, is the setting in which we obtain increasingly frequent observations of the same path. The consistency of the discrete-observation signatures as approximations of the continuous signature is also considered in \cite{lucchese2025learning}, where it is shown that, assuming the observation process is sufficiently regular (in particular, it can be seen as a `geometric rough path'), then the estimates are consistent, and convergence rates are given in terms of the frequency of observations.

In practice, this is less critical a concern for our analysis, as we do not suppose that it is generally possible to monitor the observation variables with arbitrary frequency. This aligns with the view that the observation process is not under the direct control of the econometrician, but is determined by other factors (such as availability of survey data, frequency of corporate reporting, etc.). Nevertheless, it is reassuring to know that there is indeed a consistency result in this direction, particularly when observations in high-frequency are available.